\documentclass[10pt,twocolumn,twoside]{IEEEtran}
\ifCLASSINFOpdf
\else
\fi

\usepackage{url}
\usepackage{verbatim}

\usepackage{amsmath} 
\usepackage{amssymb}
\usepackage{algorithm}
\usepackage{tikz}
\usepackage{amsthm}
\usepackage{romannum}

\usepackage{caption}
\usepackage{subcaption}

\usepackage{booktabs}
\usepackage{tabularx}
\newcolumntype{C}{>{\centering\arraybackslash}X}

\usepackage[justification=centering]{caption} 

\usepackage{pgfplots}
 \usepackage{pgfplotstable}
\pgfplotsset{compat = newest}

\newtheorem{theo}{\noindent Theorem}
\newtheorem{problem}[theo]{Problem}

\newtheorem{lemma}[theo]{Lemma}
\newtheorem{remark}[theo]{Remark}

\newtheorem{definition}[theo]{Definition}

\newtheorem{example}[theo]{Example}

\begin{document}
%
\title{Asynchronous Distributed Consensus with Minimum Communication}
%
%
%

\author{Vishal~Sawant,
        Debraj~Chakraborty
        and~Debasattam~Pal
\thanks{Vishal Sawant, Debraj Chakraborty (Corresponding author) and Debasattam Pal are with 
the Department of Electrical Engineering, Indian Institute of Technology Bombay, Mumbai - 400076, India.
{\tt\small Email: \{vsawant, dc, debasattam\}@ee.iitb.ac.in}} }
\maketitle

\begin{abstract}
In this paper, the communication effort required in a multi-agent system (MAS) is minimized via an explicit optimization formulation. The paper considers a MAS of single-integrator agents with bounded inputs and a time-invariant communication graph. A new model of discrete asynchronous communication and a distributed consensus protocol based on it, are proposed. The goal of the proposed protocol is to minimize the aggregate number of communication instants of all agents, required to steer the state trajectories inside a pres-specified bounded neighbourhood within a pre-specified time. Due to information structure imposed by the underlying communication graph, an individual agent does not know the global parameters in the MAS, which are required for the above-mentioned minimization. To counter this uncertainty, the worst-case realizations of the global parameters are considered, which lead to min-max type optimizations. The control rules in the proposed protocol are obtained as the closed form solutions of these optimization problems. Hence, the proposed protocol does not increase the burden of run-time computation making it suitable for time-critical applications. 
\end{abstract}


%
\IEEEpeerreviewmaketitle

\section{Introduction}\label{Introduction}
In the recent years, multi-agent systems (MASs) have received tremendous attention due to their extensive applications
in unmanned aerial vehicles (UAVs) \cite{Golias}, sensor networks \cite{Hla}, power grids \cite{Merabet}, industrial robotics \cite{Shamma} etc. 
A typical MAS consists of a group of agents which collaborate to achieve desired objectives.
Often, the objective is to achieve \textit{consensus}, i.e., to drive
the states of all agents into an agreement. 
To achieve this objective, agents in MAS need to exchange information with each other over some communication network.
It is well known that communication network is an expensive resource \cite{Antsaklis}, \cite{Nair1}.
Hence, with the intent of reducing communication effort, 
a few approaches have been proposed in \cite{Dimarogonas}, \cite{WenLinear},  \cite{WenSecond} etc.
However, except in our preliminary work \cite{Sawant}, the minimization of communication effort via explicit
optimization, 
has never been addressed.

The problem of achieving consensus of MAS with single-integrator agents was first extensively analyzed in \cite{Jadbabaie}.
After that, various consensus protocols for higher order agents were developed in \cite{Fax}, \cite{Ren}, \cite{Ren_Beard} etc.
These protocols are either based on discrete \textit{synchronous} communication or on continuous communication.
A global synchronization clock is necessary for the implementation of synchronous protocols, which 
is often a major practical constraint \cite{Joshi}.
On the other hand,
the energy consumption of agent's transponder is proportional to the number and duration of transmissions \cite{Corral}. 
Thus, continuous transmission limits the life of agent's battery
and hence, its flight time \cite{Bertran}, \cite{Meyer}.
Second, 
continuous transmissions over a shared bandwidth-limited channel by multiple agents
may lead to congestion of communication channel \cite{Antsaklis}, \cite{Nair1}. 
And finally, MASs such as a group of UAVs are often used for stealthy military applications \cite{Liu}, \cite{Zhao}.
In such scenarios, it is of strategic advantage to keep radio transmissions to a minimum in order to avoid detection by the enemy.
For all these applications, it is necessary to develop an asynchronous/intermittent communication based consensus protocol which minimizes communication 
effort, i.e., the number and/or duration of transmissions.

In order to reduce the required communication effort in a MAS, a few indirect approaches have been proposed in the literature.
In self-triggered control \cite{Tabuada1} based consensus protocols \cite{Johansson1}, \cite{Johansson2},
the next communication instant is pre-computed based on the current state.  
Event-triggered control \cite{Tabuada1} based consensus protocols \cite{Dimarogonas}, \cite{Zhu} initiate communication
only when a certain error state reaches a predetermined threshold.
Intermittent communication based consensus protocols are investigated
in \cite{WenLinear}, \cite{WenSecond} and \cite{WenNonlinear}.
Consensus protocols based on asynchronous information exchange for a MAS with single-integrator agents are developed in
\cite{Cao_discrete},
\cite{Cao}, \cite{Fang} and \cite{Xiao}.
Other work on asynchronous consensus include \cite{Almeida}, \cite{Giannini}, \cite{Cao_sarysmakov} and \cite{Zhan}.
All of the above protocols result in the reduced communication effort as compared to the conventional continuous communication based protocols.
However, there is no explicit minimization of communication effort and hence, the above protocols can result in sub-optimal communication performance.

To overcome this issue, in this paper, we develop a distributed protocol which minimizes the communication effort required to maintain the consensus of single-integrator agents, via explicit optimization. This protocol is based on discrete asynchronous communication. Our notion of consensus is less stringent than the conventional one \cite{Jadbabaie}, in that we only require 
the difference between neighbouring agents' states (i.e., the \textit{local disagreement})
 to reduce below a pre-specified bound.
In the proposed protocol, communication occurs at discrete time instants, namely \textit{update/communication instants}, at which agents 
access the states of neighbouring agents and then based on that information, update their control.
Thus, in the proposed protocol, the number of update instants is
a good measure of communication effort. Hence, our basic objective is to 
minimize the aggregate number of update instants of all agents in the MAS.
Now, intuitively, if the inter-update durations are increased, 
then the number of update instants decreases. However, increase in the inter-update durations 
increases the time required to achieve 
the consensus bound.
Hence, the problem of minimization of the number of update instants is well-posed only when the consensus time
is included in the formulation.
Evidently, the time required for a MAS to achieve consensus depends on the initial configuration of the agents.
Hence, to make the minimization of the number of update instants well-posed,
we require that the local disagreements of all agents be steered below a pre-specified consensus bound
within a pre-specified time, which is expressed as a function of the initial condition of the MAS.

Due to communication structure imposed by the network, 
the above-mentioned minimization is a \textit{decentralized} optimal control \cite{Movric} problem.
Because of the said imposition, an individual agent does not have global information such 
as the number of agents in the MAS, the underlying communication graph, the complete initial condition of the MAS etc.
To counter this lack of global information, we require that the consensus constraint be 
satisfied for any number of agents, any communication graph and any initial condition.
Further, due to the imposed communication structure, 
an individual agent can not predict the control inputs of the neighbouring agents.
In order to guard against the resulting uncertainty, the control inputs in the proposed protocol are 
obtained as a solution of 
certain max-min optimizations (Sections \ref{z_inside} and \ref{z_outside}). 
We obtain the closed form solutions of these optimizations (see (\ref{update_instant})-(\ref{control_out})) and 
hence, extensive numerical computations are not required for their implementation.
This makes the proposed protocol suitable for time-critical applications.

Our contributions in this paper are summarized as follows:
\begin{enumerate}
 \item We develop a discrete asynchronous communication based distributed consensus protocol
 for a MAS with single-integrator agents (Section \ref{proposed}).
 \item The proposed protocol minimizes the aggregate number of update instants under the constraint of steering the local disagreements of all agents below a pre-specified limit
 within a pre-specified time
 (Theorems \ref{consensus_time_larger} and \ref{proposed_optimal_larger}).
 \item 
 The control rules in the proposed protocol are solution of certain max-min optimizations (Sections \ref{z_inside} and \ref{z_outside}).
       We obtain the closed form expressions of the corresponding optimal controls (Lemmas \ref{opt_inside} and \ref{opt_outside}).
\end{enumerate}
 The current paper is an extension of our work in  \cite{Sawant} in three major ways. First, it was assumed in \cite{Sawant} that
the initial local disagreements between agents are confined below a pre-specified bound.
In the current paper, no such assumption on initial conditions  has been made. Second, the protocol in \cite{Sawant} solves the minimization problem for a specific consensus time. On the other hand, the protocol developed in the current paper solves the minimization problem for a general pre-specified consensus time. Finally, in the current paper, the effect of pre-specified consensus time on optimal communication cost is analyzed. Such analysis was not presented in \cite{Sawant}.

The remaining part of this paper is organized as follows. In Section \ref{prob_formulation}, the problem of minimizing the number of communication instants is formulated. In Section \ref{proposed}, a distributed consensus protocol is proposed, which will be shown to be the
solution for the special case of the formulated problem, in Sections \ref{consensus_under_our} and \ref{Optimality_2T}. The protocol proposed in Section \ref{proposed} is extended for the general case of the formulated problem, in Section \ref{T_larger}. In Section \ref{simulation}, the simulation results are presented. The paper is concluded in Section \ref{conclusion1} with future directions.

\section{Preliminaries and Problem formulation}\label{prob_formulation}
\subsection{Graphs}
A graph $G=(V,E)$ is a finite set of nodes $V$ 
connected by a set of edges  $E \subseteq (V\times V)$. An edge between nodes $i$ and $j$ is represented by an ordered pair $(i,j) \in E$. A graph $G$ is said to be \textit{simple} if $(i,i)\not \in E,~\forall i \in V$. A graph $G$ is said to be \textit{undirected} if $(i,j) \in E$ implies $(j, i) \in E$.
In an undirected graph $G$, if $(i,j) \in E$ (and equivalently $(j,i) \in E$), then the nodes $i$ and $j$ are said to be \textit{neighbours} of each other.
A \textit{path} between nodes $i$ and $j$ in an undirected graph $G$ is a sequence of edges $(i,k_1), (k_1,k_2),\dots,(k_{r-1},k_r),(k_r,j) \in E$. An undirected graph $G$ is said to be \textit{connected} if there exists a path between any two nodes in $G$.
Let $n_i$ denote the number of neighbours of node $i$ and $|V|$ denote the cardinality of set $V$.
Then, the \textit{Laplacian} matrix $L \in {\mathbb{R}}^{|V| \times |V|}$ of a simple undirected graph $G=(V,E)$ is defined as
\begin{equation*}
 L_{i,j}: = \begin{cases}
                        n_i,& \text{ if } ~~~~i=j \\
                       -1,& \text{ if } ~~~~i \neq j ~~\text{ and }~~ (i,j) \in E\\
                        0,& \text{ if } ~~~~i \neq j ~~\text{ and }~~ (i,j) \not \in E
                      \end{cases}
\end{equation*}

\subsection{System description}\label{system}
Consider a multi-agent system (MAS) of $n$ single-integrator agents, labeled as $a_1$, $a_2, \dots, a_n$, with dynamics
\begin{equation}\label{MAS}
 \dot{x}_i(t)=u_i(t),~~~~~i = 1,\dots, n
\end{equation}
where $x_i(t) \in \mathbb{R}$ and $u_i(t) \in \mathbb{R}$ are the state and the control input of agent $a_i$, respectively.
Let $X:=[x_1,x_2,\dots,x_n]$ and ${\bf{u}}:=[u_1,u_2,\dots,u_n]$ be the augmented state and control vectors of MAS (\ref{MAS}), respectively.
Define the set $P:=\{1,2,\dots,n\}$.

Let $G$ be a \textit{time-invariant}  simple undirected graph, whose nodes
represent agents in MAS (\ref{MAS}) whereas the edges represent the communication 
links between agents, over which they exchange information with their neighbours.
Let $S_i$ be the set of indices of neighbours of agent $a_i$. Note that $i \not \in S_i$ as $G$ is a simple graph.
The cardinality of the set $S_i$ is denoted by $n_i$.
Let $L$ denote the \textit{Laplacian} matrix of $G$.

We make the following assumptions about MAS (\ref{MAS}):
\begin{enumerate}
 \item The control input $u_i$ of each agent belongs to the set
 \begin{equation*}
  \mathcal{U}:=\{ u \in \mathcal{M}~|~|u(t)| \leq \beta,~\forall t\geq 0\} 
 \end{equation*}
 where $\mathcal{M}$ denotes the set of measurable functions from $[0,\infty)$ to   $\mathbb{R}$.
 \item The communication graph $G$ is connected.
 \item \label{delay} The communication delay is zero.
\end{enumerate}

\subsection{Consensus}
In this paper, we will be using two notions of consensus, which we define next.
\begin{definition}
 MAS (\ref{MAS}) is said to have achieved \textit{conventional consensus} at instant $\widetilde{t}$ if
 \begin{equation*}
 \widetilde{t}:=\inf \big\{\widehat{t}~ \big|~ x_i(t)=x_j(t),~~\forall t\geq \widehat{t},~~\forall i,j\in P \big\} < \infty
 \end{equation*}
\end{definition}
Define $Z(t)=[z_1(t),\dots,z_n(t)]:=LX(t)$. Then, it follows from the definition of the Laplacian matrix that
\begin{equation}\label{local_dis}
 z_i(t)=\sum_{j\in S_i} \big(x_i(t)-x_j(t)\big),~~~~~\forall i \in P
\end{equation}
As $z_i$ is the sum of differences of agent $a_i$'s state with its neighbours, we call it the 
\textit{local disagreement} of agent $a_i$.
It is well known \cite{Jadbabaie} that for a MAS with connected, time-invariant communication graph,
conventional consensus at instant $\widetilde{t}$ is equivalent to $z_i(t)=0,~\forall t\geq \widetilde{t},~\forall i\in P$.
However, in many practical applications, it is not necessary that each $z_i$ becomes exactly zero.
It is sufficient if each $|z_i|$ remains below a prespecified consensus bound. This motivates our next notion of consensus, namely $\alpha$-\textit{consensus}. 
\begin{definition}\label{alpha}
 Let $\alpha \in {\mathbb{R}}^+$ be the prespecified consensus bound. MAS (\ref{MAS}) is said to have achieved
 $\alpha$-consensus at instant $\widetilde{t}$ if
 \begin{equation*}
  \widetilde{t}:=\inf \big\{\widehat{t}~ \big|~ |z_i(t)| \leq \alpha,~~\forall t\geq \widehat{t},~~\forall i \in P \big\}< \infty
 \end{equation*}
\end{definition}

\subsection{Communication model}\label{comm_model}
In this paper, we consider a discrete communication model.
Let $t^k_i$ denote the $k$th update instant (also referred to as the communication instant) of agent $a_i$, at which 
it accesses the state information of its neighbours and then, based on that information,
updates its control. 
Our communication model is asynchronous, i.e., the update
instants $t^k_i$'s of two different agents need not coincide.

As the communication model is discrete, the number of update instants
is a good measure of communication effort.
Hence, we define the \textit{communication cost} of agent $a_i$, denoted by
$C_i(t)$, as the number of update instants of agent $a_i$ in the time interval $[0,t]$.
Then, we define the \textit{aggregate} communication cost of MAS (\ref{MAS}), denoted by $C_{MAS}(t)$, as
\begin{equation}\label{agg_cost}
 C_{MAS}(t):=\sum_{i\in P} C_i(t)
\end{equation}



\subsection {Problem formulation}
Let $X(0)$ be an initial condition of MAS (\ref{MAS}), $\alpha$ be the prespecified consensus bound and 
$T$ be the prespecified consensus time.
Our objective is to develop a protocol which minimizes 
the communication cost
$C_{MAS}(T)$,
under the constraint of
achieving $\alpha$-consensus of MAS (\ref{MAS}) within time $T$. 

As discussed in Section \ref{Introduction}, due to  information structure imposed by graph $G$, 
an individual agent in MAS (\ref{MAS}) has access only to its own information and that of its neighbours. 
Therefore, the proposed protocol needs to be \textit{distributed}, i.e., based only on the local information.
In addition, an individual agent does not have global information
such as the number of agents $n$ in MAS (\ref{MAS}), the structure of the communication graph $G$,
the complete initial condition $X(0)$ etc.
To address these uncertainties, 
the proposed protocol must be able to
achieve $\alpha$-consensus of MAS (\ref{MAS}) within the prespecified time $T$,
for any $n$, any connected $G$ and any $X(0) \in {\mathbb{R}}^n$.

Since the input set $\mathcal{U}$ is magnitude bounded, for a fixed $T$, it will not be possible to achieve $\alpha$-consensus
within time $T$, for every
$X(0) \in {\mathbb{R}}^n$.
Thus, the consensus time must be specified as a function of $X(0)$.
To highlight this dependence on $X(0)$, we modify the notation 
of the prespecified consensus time from $T$ to $T\big(X(0)\big)$.
Similarly, the communication costs $C_i$ and $C_{MAS}$ depend on $X(0)$.
Thus, we modify their notations from $C_i(t)$ and $C_{MAS}(t)$ to
$C_i\big(t,X(0)\big)$ and $C_{MAS}\big(t,X(0)\big)$, respectively.
Now, we formalize our objective as follows:
\begin{problem}\label{prob_main}
 Consider MAS (\ref{MAS}) with initial condition $X(0) \in {\mathbb{R}}^n$ and  connected communication graph $G$.
Let $\Psi(n)$ denote the set of connected graphs with $n$ nodes.
Let $T\big(X(0)\big)$ be the prespecified consensus time 
 which is expressed as a function of $X(0)$.
Develop, if possible, a discrete asynchronous communication based protocol, i.e., admissible control ${\bf{u}}^{\ast}=[u_1^{\ast},\dots,u^{\ast}_n]$,
adhering to graph $G$,
 which is solution of the following optimization:
 \begin{align}
  {\bf{u}}^{\ast}=~&\min_{ \substack{u_i \in \mathcal{U} \\ \forall i \in P} }~~C_{MAS}\Big(T\big(X(0)\big),X(0) \Big) \nonumber \\
                 & \text{ s.t. }~~~|z_i(t)| \leq \alpha,~~\forall t\geq T\big( X(0) \big), \label{constraint_main} ~~ \forall i \in P\\ \nonumber 
                 & \hspace*{2.03cm} \forall n,~~\forall G \in \Psi(n),~~\forall X(0) \in {\mathbb{R}}^n
 \end{align}
\end{problem}

\subsection{Choosing $T\big(X(0)\big)$ }\label{solution_flow}
In Problem \ref{prob_main}, the time $T\big(X(0)\big)$ can be specified as any function of $X(0)$. However, in  practical applications, it is desirable to set $T\big(X(0)\big)$ to the minimum feasible value.

In \cite{Mulla_single_integrator}, the time-optimal control rule is proposed which achieves conventional consensus of MAS (\ref{MAS}) in minimum time.
This control rule and the corresponding consensus time, denoted by $u^{\star}_i$ and $T^{\ast}\big(X(0)\big)$, respectively, are presented below.
Let $X(0)=[x_1(0),\dots,x_n(0)]$ be an initial condition of MAS (\ref{MAS}). Define $x^{min}(0):=\min \big \{ x_1(0),\dots,x_n(0) \big\}$ and $x^{max}(0):=\max \big \{ x_1(0),\dots,x_n(0) \big\}$. Recall that $z_i$ denote the local disagreement of agent $a_i$. Let $sign$ denote the standard signum function. Then, the time-optimal consensus rule from \cite{Mulla_single_integrator} is 
\begin{equation}\label{opt_control}
 u^{\star}_i(t)=-\beta\hspace*{0.1cm} sign \big(z_i(t)\big),~~~~\forall t\geq 0,~~~~\forall i \in P 
 \end{equation}
 with the corresponding consensus time
 \begin{equation}\label{min_time_expression}
 T^{\ast}\big(X(0)\big)=\frac{x^{max}(0) - x^{min}(0)}{2\beta} 
\end{equation}

Notice that the control rule (\ref{opt_control}) requires instantaneous access to the local disagreement $z_i$, and as a result, demands continuous communication. Then, it follows from the time optimality of $T^{\ast}\big(X(0)\big)$ that a discrete communication based protocol cannot achieve $\alpha$-consensus of MAS (\ref{MAS}) within time $T^{\ast}\big(X(0)\big)$. This implies that Problem \ref{prob_main} is infeasible for $T\big(X(0)\big)\leq T^{\ast}\big(X(0)\big)$ and we should assume $T\big(X(0)\big)> T^{\ast}\big(X(0)\big)$.
We further assume that $T\big(X(0)\big) \geq 2T^{\ast}\big(X(0)\big)$. This assumption makes Problem \ref{prob_main} tractable and results in a particularly simple closed form solution for the control inputs.
  We solve Problem \ref{prob_main} for $T\big(X(0)\big)= 2T^{\ast}\big(X(0)\big)$ in Sections \ref{proposed}-\ref{Optimality_2T} and later extend it for $T\big(X(0)\big)> 2T^{\ast}\big(X(0)\big)$ in Section \ref{T_larger}.

\section{Protocol $\Romannum{1}$: For $ T\big(X(0)\big) = 2T^{\ast}\big(X(0)\big)$}\label{proposed}
In this section, we present the protocol proposed for $ T\big(X(0)\big) = 2T^{\ast}\big(X(0)\big)$.
We refer to this protocol as Protocol $\Romannum{1}$. Later, this protocol will be shown to be a solution of Problem \ref{prob_main} for
$ T\big(X(0)\big) = 2T^{\ast}\big(X(0)\big)$. 

Protocol $\Romannum{1}$ has the following two elements:\\
\textit{$1$)~ Computation of next update instant and control input }\\
$a$)~ At each update instant $t^k_i,~k\geq1$, agent $a_i,~ i\in P$, accesses
the current states $x_j(t^k_i)$'s of its neighbours $a_j,~ j \in S_i$, 
and computes $z_i(t^k_i) = \sum_{j \in S_i} \big(x_i(t^k_i)-x_j(t^k_i)\big)$.\\
$b$)~After that, agent $a_i$ computes its next update instant $t^{k+1}_i$  
and control input $u^{\ast}_i$ to be applied in the interval $[t^k_i,t^{k+1}_i)$ as follows:\\
\hspace*{0.5cm}$\romannum{1}$) If $|z_i(t^k_i)| \leq \alpha$, then
\begin{align}
 t^{k+1}_i &= t^k_i + \frac{\alpha}{\beta n_i} \label{update_instant}\\
 u^{\ast}_i(t) &= -\frac{z_i(t^k_i)}{\alpha}\beta,\hspace*{0.9cm} \forall t\in \big[t^k_i,t^{k+1}_i\big) \label{control}
\end{align}
\hspace*{0.5cm}$\romannum{2}$) If $|z_i(t^k_i)| > \alpha$, then
\begin{align}
 t^{k+1}_i &= t^k_i + \frac{|z_i(t^k_i)|+\alpha}{2\beta n_i}\label{update_instant_out} \\
 u^{\ast}_i(t) &= -\beta sign\big (z_i(t^k_i) \big),~~~~\forall t\in \big[t^k_i,t^{k+1}_i\big) \label{control_out}
\end{align}
$c$)~Then, agent $a_i$ broadcasts $x_i(t^k_i)$ and $u^{\ast}_i(t^k_i)$.
This broadcast information is received by agent $a_i$'s neighbours $a_j,~j\in S_i$, at the same instant 
$t^k_i$.
The neighbours $a_j,~j\in S_i$, store this information with the reception time-stamp.

\textit{$2$) Accessing the states of neighbours at update instants }\\
$a$)~The time instant $t^1_i=0$ is the first update instant of all agents $a_i,~i \in P$. 
At this instant, all agents broadcast their current states.
Thus, each agent $a_i,~i\in P$, has a direct access to the current states $x_j(t^1_i)$'s of its neighbours $a_j,~j \in S_i$.
For example, consider the communication graph $G_c$ shown in Fig. \ref{graph_in_model}
and the corresponding communication timeline shown in Fig. \ref{comm_timeline}.
At instant $t=0$, agents $a_1$, $a_2$ and $a_3$ broadcast information to their neighbours.
\begin{figure}[h]
\begin{center}
\begin{tikzpicture}[scale=0.3]
 \draw (-1,0) circle [radius=1.2];
 \draw (-10,0) circle [radius=1.2];
 \draw (8,0) circle [radius=1.2];
 \draw (-2.2,0) to (-8.8,0);
 \draw (0.2,0) to (6.8,0);
 \node at (-10,0) {$1$};
 \node at (-1,0) {$2$};
 \node at (8,0) {$3$};
\end{tikzpicture}
\end{center}
\caption{Communication graph $G_c$ }\label{graph_in_model}
\end{figure}
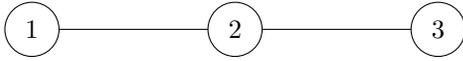
\begin{figure}[h]
\begin{center}
\begin{tikzpicture}[scale=0.3]
\draw [->] (-16,0) to (6,0);
\draw [->] (-16,-4) to (6,-4);
\draw [->] (-16,-8) to (6,-8);
\draw[thick,<->] (-16,-0.1) to (-16,-3.9);
\draw[thick,<->] (-16,-4.1) to (-16,-7.9);
\draw[thick,->] (-12,0) to (-12,-4);
\node at (-12,0) [circle,fill,inner sep=1.5pt]{};
\draw[dashed] (-12,-4) to (-12,-8);
\draw[thick,->] (-7,-8) to (-7,-4);
\node at (-7,-8) [circle,fill,inner sep=1.5pt]{};
\draw[thick,->] ( -3,-3.9) to (-3,0);
\draw[thick,->] ( -3,-4.1) to (-3,-8);
\node at (-3,-4) [circle,fill,inner sep=1.5pt]{};
\draw[thick,->] (3.1,0) to (3.1,-4);
\node at (3.1,0) [circle,fill,inner sep=1.5pt]{};
\draw[dashed] (3.1,-4) to (3.1,-8);
\draw[] (-16,0.2) to (-16,-0.2);
\draw[] (-16,-3.8) to (-16,-4.2);
\draw[] (-16,-7.8) to (-16,-8.2);
\node at (-15.9,-9.2) {$0$};
\node at (-11.9,-9.2) {$t^k_1$};
\node at (-6.9,-9.2) {$t^l_3$};
\node at (-2.9,-9.2) {$t^p_2$};
\node at (3.4,-9.2) {$t^{k+1}_1$};
\node at (-18,-0.1) {$a_1$};
\node at (-18,-4.1) {$a_2$};
\node at (-18,-8.1) {$a_3$};
\end{tikzpicture}
\end{center}
\caption{Timeline of communication over graph $G_c$ in Fig. \ref{graph_in_model}. Dark arrows indicate information transmissions.
}\label{comm_timeline}
\end{figure}
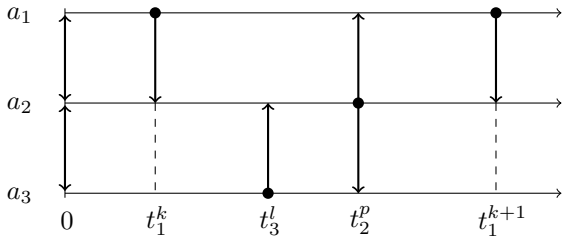

$b$)~Let $t^k_i,~k>1$, be any update instant of agent $a_i$. Let $t^l_j<t^k_i$ be the latest update instant of agent $a_j,~j\in S_i$,
at which it had broadcast $x_j(t^l_j)$ and $u^{\ast}_j(t^l_j)$.
At update instant $t^k_i$, agent $a_i$ accesses the stored information and retrieves $x_j(t^l_j)$ and $u^{\ast}_j(t^l_j)$ for all $j \in S_i$.
For example, in the communication timeline shown in Fig. \ref{comm_timeline},
at instants $t^k_1$ and $t^l_3$, agent $a_2$ receives information from its neighbours. 
Later, agent $a_2$ uses this information at its update instant $t^p_2$.\\
$c$)~It is known from (\ref{control}) and (\ref{control_out}) that every agent $a_j,~j \in S_i$, had applied control $u^{\ast}_j(t)=u^{\ast}_j(t^l_j)$ in the interval $ t \in [t^l_j,t^k_i)$.
Using this fact, agent $a_i$ computes the current state $x_j(t^k_i)$ of $a_j,~\forall j\in S_i$, as
\begin{equation*}\label{neighbour_state}
 x_j\big(t^k_i\big)=x_j(t^l_j) + \int_{t^l_j}^{t^k_i}\!u^{\ast}_j\big(t^l_j\big)~ dt
\end{equation*}

\begin{remark}\label{computation_fast}
 As per Assumption \ref{delay}, the communication delay is zero.
 In addition, the time required for the computation of $t^{k+1}_i$ and $u^{\ast}_i(t^k_i)$
 is negligible.
 This justifies the assumption that computation, transmission and reception of $u^{\ast}_i(t^k_i)$ and 
 $x_i(t^k_i)$, happen at the same time
 instant $t^k_i$.
\end{remark}

\begin{remark}\label{control_broadcast}
 It may appear from (\ref{control}) and (\ref{control_out}) that we have selected $u^{\ast}_i$'s which are constant over
 intervals $[t^k_i,t^{k+1}_i)$, in order to simplify analysis. 
 However, $u^{\ast}_i$'s in (\ref{control}) and (\ref{control_out}) 
 are obtained as the solution of certain max-min optimizations (Sections \ref{z_inside} and \ref{z_outside})
 and coincidently, they have this nice form. 
\end{remark}

\section{$\alpha$-consensus under Protocol $\Romannum{1}$}\label{consensus_under_our}
In this section, we show that Protocol $\Romannum{1}$
achieves $\alpha$-consensus of MAS (\ref{MAS}) within time  $2T^{\ast}\big(X(0)\big)$.
However, before that, we present a necessary condition on feasible solutions of Problem \ref{prob_main}, which will be utilized while proving the attainment of $\alpha$-consensus within time $2T^{\ast}\big(X(0)\big)$.
\subsection{Necessary condition on feasible solutions of Problem \ref{prob_main}}
A discrete communication based distributed protocol is said to be a \textit{feasible} solution of Problem \ref{prob_main}
if it satisfies constraint (\ref{constraint_main}), 
i.e., achieves $\alpha$-consensus of MAS (\ref{MAS}) within time $T\big(X(0)\big)$, for
any number of agents $n$, any connected graph $G$ and any initial condition 
$X(0)\in {\mathbb{R}}^n$. Recall that $\Psi(n)$ denotes the set of connected graphs with $n$ nodes. Then, the following lemma gives a necessary condition on the feasible solutions of Problem \ref{prob_main}.

\begin{lemma}\label{necessary_condition}
Consider MAS (\ref{MAS}).
 Let $T\big(X(0)\big)$ be the pre-specified consensus time.
 Let $Q$ be any discrete communication based distributed protocol.
 Under Protocol $Q$, let $\widetilde{t}_i$ be the first time instant at which the local disagreement $z_i$ of agent $a_i$
 satisfies $\big|z_i\big(~\!\widetilde{t}_i\big)\big| \leq \alpha$. Then,
 Protocol $Q$ is a feasible solution of Problem \ref{prob_main} only if it ensures
 \begin{equation}\label{retain_inside}
  \big|z_i(t)\big| \leq \alpha,\hspace*{0.4cm}\forall t\geq \widetilde{t}_i,\hspace*{0.2cm} \forall n, \hspace*{0.2cm} \forall G \in \Psi(n), \hspace*{0.2cm} \forall X(0) \in {\mathbb{R}}^n
 \end{equation}
\end{lemma}
\begin{proof}
 Recall that due to  communication structure imposed by graph $G$, 
 an individual agent $a_i$ in MAS (\ref{MAS}) does not know the complete initial condition $X(0)$. Therefore, it does not know 
 the exact value of $T\big(X(0)\big)$
 and how far $\widetilde{t}_i$ is from
 $T\big(X(0)\big)$. 
 In fact, as in Example \ref{strict_bound} presented below, there may exist a MAS of the form (\ref{MAS})
 in which $\widetilde{t}_i=T\big(X(0)\big)$. 
 In such a case, violation of
 (\ref{retain_inside}) results in
 the violation of constraint (\ref{constraint_main}) in Problem \ref{prob_main}. This contradicts the fact that Protocol $Q$ is a feasible solution
 of Problem \ref{prob_main}. Hence, Protocol $Q$ must satisfy
 (\ref{retain_inside}).
 This completes the proof.
\end{proof}



The following lemma shows that Protocol $\Romannum{1}$ satisfies the necessary condition (\ref{retain_inside}).
The proof of this lemma relies on the derivation of control rules (\ref{update_instant})-(\ref{control_out}) in Protocol $\Romannum{1}$, which is deferred to Section \ref{Optimality_2T} for better structure of the paper.  Hence, we defer the proof of the said lemma to Section \ref{feasible1}.
\begin{lemma}\label{feasible}
Protocol $\Romannum{1}$ satisfies the necessary condition (\ref{retain_inside}). 
\end{lemma}

\subsection{$\alpha$-consensus within time $2T^{\ast}\big(X(0)\big)$}
The following theorem shows that Protocol $\Romannum{1}$
achieves $\alpha$-consensus of MAS (\ref{MAS}) within time $2T^{\ast}\big(X(0)\big)$.
\begin{theo}\label{consensus_time}
 Consider MAS (\ref{MAS}).
 Let $\alpha \in {\mathbb{R}}^{+}$
 be the pre-specified consensus bound.
 Let 
 $T^{\ast}\big(X(0)\big)$ be as defined in (\ref{min_time_expression}).
 Then, for every $n$, every connected communication graph $G$ and every $X(0) \in {\mathbb{R}}^n$, Protocol $\Romannum{1}$
 achieves $\alpha$-consensus of MAS (\ref{MAS}) in time less than or equal to $2T^{\ast}\big(X(0)\big)$. 
\end{theo}
\begin{proof}
See the Appendix for the proof.
\end{proof}

According to Theorem \ref{consensus_time}, under Protocol $\Romannum{1}$,
the duration $2T^{\ast}\big(X(0)\big)$ is an upper bound on the time required 
to achieve $\alpha$-consensus of MAS (\ref{MAS}). Next, we show with the following example that there exists a MAS of the form (\ref{MAS}) for which the $\alpha$-consensus time under Protocol $\Romannum{1}$ is arbitrarily close to $2T^{\ast}\big(X(0)\big)$.

\begin{example}\label{strict_bound}
Consider the connected graph $G_e$ shown in Fig. \ref{graph_in_example}, which has $n$ nodes.
The subgraph $G_s$ of $G_e$ (inside the dashed square in Fig. \ref{graph_in_example}) is a complete graph on $n-1$ nodes.
The set of indices of neighbours of node $1$ is $S_1=\{2,3,\dots,r+1\}$.
Hence, the cardinality of $S_1$ is $r$.

Now, consider MAS (\ref{MAS}) with communication graph $G_e$.
Let the consensus bound and the control bound be $\alpha=3$ and $\beta=1$, respectively.
Let the initial conditions of agents be $x_1(0)=0$ and $x_i(0)=5,~\forall i=2,\dots,n$.
Thus, $x^{min}(0)=\min_{i} x_i(0)=0$ and $x^{max}(0)=\max_{i} x_i(0)=5$.
Then, it follows from the definition of $T^{\ast}\big(X(0)\big)$ in (\ref{min_time_expression}) that 
$T^{\ast}\big(X(0)\big)=2.5$ sec.
\end{example}

\begin{figure}[ht]
\begin{center}
\begin{tikzpicture}[scale=0.4]
\draw (-1,0) circle [radius=0.9];
\draw (1.7,5.49) circle [radius=0.9];
\draw (7.4,6.9) circle [radius=0.9];
\draw (13,0) circle [radius=0.9];
\draw (1.7,-5.49) circle [radius=0.9];
\draw (7.4,-6.9) circle [radius=0.9];

\draw (-1,0.9) to [out=80,in=225] (1,5);
\draw[thick,dotted] (2.4,6) to [out=29,in=177] (6.5,7);

\draw (-1,-0.9) to [out=280,in=135] (1,-5);
\draw [dotted, thick] (2.4,-6) to [out=331,in=179] (6.5,-7);

\draw [dotted,thick] (8.3,6.6) to [out=340,in=100] (13,0.9);
\draw [dotted,thick] (8.3,-6.6) to [out=20,in=260] (13,-0.9);

\draw (-7,0) circle [radius=0.9];
\draw (-6.1,0) to (-1.9,0);
\draw (-6.1,0) to (0.97,5);
\draw (-6.1,0) to (0.97,-5);
\draw (-6.1,0) to [out=71,in=150] (6.5,7);
\draw (-6.1,0) to [out=289,in=210] (6.5,-7);

\draw (-0.1,0) to (7.3,6);
\draw (-0.1,0) to (12.1,0);


\draw (7.3,6) to (12.1,0);
\draw (7.3,-6) to (12.1,0);

\draw (12.1,0) to (2.1,4.64);
\draw (12.1,0) to (2.1,-4.64);

\draw (2.1,4.64) to (7.3,-6);
\draw (2.1,-4.64) to (7.3,6);

\draw (2.1,4.64) to (2.1,-4.64);
\draw (7.3,6) to (7.3,-6);

\draw [dashed, thick] (-2.5,8.5)--(14.6,8.5)--(14.6,-8.5)--(-2.5,-8.5)--(-2.5,8.5);

 \node at (-1,0) {$2$};
 \node at (1.7,5.49) {$3$};
 \node at (7.4,6.9) {$r$};
 \node at (13,0) {$n$};
 \node at (6.82,-6.9) {${\footnotesize{r}}$};
 \node at (7.37,-6.9) {${\footnotesize{+}}$};
 \node at (8,-6.9) {$1$};
 \node at (1.7,-5.49) {$4$};
 \node at (-7,0) {$1$};
 \node at (12.1,6) {$G_s$};
 
\end{tikzpicture}
\end{center}
\caption{Communication graph $G_e$ of the MAS in Example \ref{strict_bound}}\label{graph_in_example}
\end{figure}
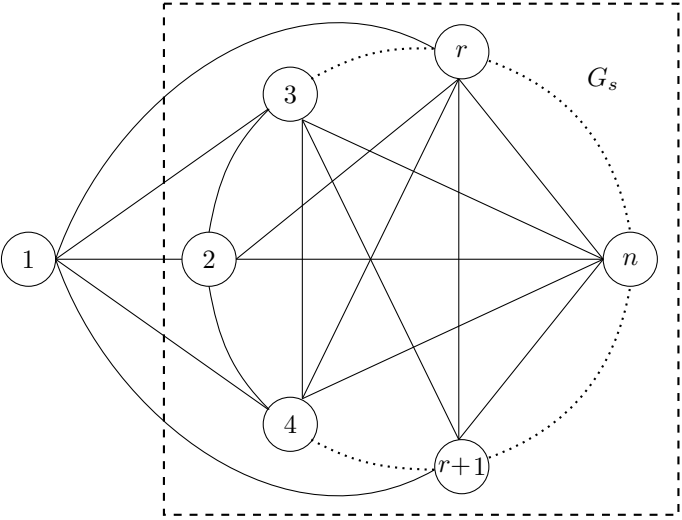

The following theorem shows that under Protocol $\Romannum{1}$, the time required to achieve $\alpha$-consensus of the MAS in Example \ref{strict_bound} is arbitrarily close to  $2T^{\ast}\big(X(0)\big)$.
\begin{theo}\label{strict_theo_1}
Consider the MAS in Example \ref{strict_bound}. Let $T\big(\alpha,X(0)\big)$ denote the time required by Protocol $\Romannum{1}$ to achieve $\alpha$-consensus of 
a MAS with initial condition $X(0)$.
Let $\epsilon>0$ be any real number. 
Then, there exist integers $r=r_{\epsilon}$ and $n=n_{\epsilon}$ such that the following holds.
\begin{equation}\label{bound_T_alpha}
   \Big(2T^{\ast}\big(X(0)\big) - \epsilon \Big) \leq  T\big(\alpha,X(0)\big)  \leq 2T^{\ast}\big(X(0)\big)
\end{equation}
\end{theo}
\begin{proof}
See the Appendix for the proof.
\end{proof}

\section{Optimality of Protocol $\Romannum{1}$}\label{Optimality_2T}
The objective in Problem \ref{prob_main} is to minimize the number of update instants under the constraint of achieving
$\alpha$-consensus of MAS (\ref{MAS}) within the pre-specified time.
Intuitively, if the duration between the successive update instants is increased, then
the number of update instants decreases. Motivated from this intuition, we obtain the solution of 
Problem \ref{prob_main} by maximizing the inter-update durations.
We divide the solution process into two steps.
First, we solve
two maximization problems, 
one corresponding to $|z_i(t^k_i)| \leq \alpha$ and the other corresponding to $|z_i(t^k_i)| > \alpha$,
in which the objective function is the inter-update duration.
The control rules
(\ref{update_instant})-(\ref{control}) and (\ref{update_instant_out})-(\ref{control_out}) in 
Protocol $\Romannum{1}$ are solutions of these maximization problems, respectively.
Later, we show how these two control rules together form the solution of Problem \ref{prob_main}.

\subsection{Maximization of inter-update durations: For $\big|z_i\big(t^k_i\big)\big| \leq \alpha$}\label{z_inside}
Let $z_i\big(t^k_i\big)$ be the local disagreement of agent $a_i$ at its update instant $t^k_i$ such that
$\big|z_i\big(t^k_i\big)\big| \leq \alpha$. The goal of agent $a_i$ is to maximize the inter-update duration $\big(t^{k+1}_i - t^k_i\big)$
by delaying its next update instant $t^{k+1}_i$. However, while doing so, agent $a_i$ must satisfy the necessary condition (\ref{retain_inside}).
For that purpose, agent $a_i$ needs to ensure that 
$|z_i(t)| \leq \alpha$ for all $t\in [t^k_i,t^{k+1}_i]$.
Recall the definition of $z_i$ in (\ref{local_dis}).
 Then, the evolution of $z_i$ in the interval $ [t^k_i,t^{k+1}_i]$ is  given as \begin{equation}\label{z_trajectory}
z_i\big(t\big) = z_i\big(t^k_i\big) + n_i \int_{t^k_i}^{t} u_i(\tau)d\tau - \sum_{j\in S_i} \bigg(\int_{t^k_i}^{t} u_j(\tau)d\tau\bigg)
\end{equation}
This evolution depends on control inputs $u_i$ and $u_j,~\forall j\in S_i$,
in the interval $[t^k_i,t^{k+1}_i]$. Note that any instant $t\in [t^k_i,t^{k+1}_i]$ can be an update instant of a neighbouring agent $a_j,~j \in S_i$, at which $a_j$ updates its control. This updated value is not known to agent $a_i$ in advance, at instant $t^k_i$.
Thus, while maximizing the inter-update duration $\big(t^{k+1}_i - t^k_i\big)$, agent $a_i$ needs to consider the worst-case 
realizations of the neighbouring inputs $u_j,~\forall j\in S_i$, which result in the minimum value of the inter-update duration. 
This leads to the following max-min optimization:
\begin{align}
u^{\ast}_i =  ~\max\limits_{ \substack{ u_i\in \mathcal{U}, \\ \widetilde{t} \in \mathbb{R}} }~& \min\limits_{ \substack{ u_j \in \mathcal{U},\\ \forall j\in S_i } } ~~~ \widetilde{t} - t^k_i \label{max_min_con}\\ 
t^{k+1}_i = \arg~ \max\limits_{ \substack{ u_i\in \mathcal{U}, \\ \widetilde{t} \in \mathbb{R}} }~&\min\limits_{ \substack{ u_j \in \mathcal{U},\\ \forall j\in S_i } } ~~~ \widetilde{t} - t^k_i \label{max_min_time}\\
 & ~~~~\text{s.t.}~~~ |z_i(t)|\leq \alpha,~~\forall t\in \big[t^k_i,\widetilde{t}\hspace*{0.1cm}\big] \label{constraint}
\end{align}

The following lemma shows that control law (\ref{update_instant})-(\ref{control}) is the solution of optimization (\ref{max_min_con})-(\ref{constraint}).
\begin{lemma}\label{opt_inside}
Consider any agent $a_i$ in MAS (\ref{MAS}).
 Let $z_i\big(t^k_i\big)$ be the local disagreement of $a_i$ at its update instant $t^k_i$ such that $\big|z_i\big(t^k_i\big)\big| \leq \alpha$.
 Then, the control law (\ref{update_instant})-(\ref{control}) is
 the solution of optimization (\ref{max_min_con})-(\ref{constraint}).
\end{lemma}
\begin{proof}
See the Appendix for the proof.
\end{proof}

\subsection{Maximization of inter-update durations: For $\big|z_i\big(t^k_i\big)\big|> \alpha$}\label{z_outside}
Let $z_i\big(t^k_i\big)$ be the local disagreement of agent $a_i$ at its update instant $t^k_i$ such that
$\big|z_i\big(t^k_i\big)\big| > \alpha$. 
Then, in order to achieve $\alpha$-consensus,  it is necessary to steer $\big|z_i\big(t^k_i\big)\big|$ below the consensus bound $\alpha$. If that could be done in time-optimal manner, it is an added advantage. Recall that in Section \ref{solution_flow}, we discussed the time-optimal consensus rule (\ref{opt_control}) which achieves conventional consensus of MAS (\ref{MAS}) in minimum time. Motivated from this rule, we choose our control rule as
\begin{equation}\label{control_out_1}
 u^{\ast}_i(t)= -\beta\hspace*{0.1cm} sign \Big(z_i\big(t^k_i\big)\Big),~~~~~~~~\forall t \in \big[t^k_i,t^{k+1}_i\big)
\end{equation}
As mentioned in Section \ref{z_inside}, the goal of agent $a_i$ is to 
maximize the inter-update duration $\big(t^{k+1}_i - t^k_i\big)$
by delaying its next update instant $t^{k+1}_i$. However, while doing so,
agent $a_i$ must satisfy the necessary condition (\ref{retain_inside}). 

Recall that $\big|z_i\big(t^k_i\big)\big| > \alpha$. Without loss of generality, assume that $z_i\big(t^k_i\big) < -\alpha$.
Then, it follows from (\ref{control_out_1}) that
$u^{\ast}_i(t)=\beta,~\forall t \in [t^k_i,t^{k+1}_i)$.
Let $\hat{t} \in [t^k_i,t^{k+1}_i)$ be the first time instant at which $z_i(\hat{t}~\!)=-\alpha$.
Then, in order to satisfy (\ref{retain_inside}), agent $a_i$ needs to ensure that
\begin{equation*}\label{z_inside_1}
 z_i(t) \in [-\alpha,\alpha],~~~~~~~~\forall t \in [\hat{t},t^{k+1}_i]
\end{equation*}
which is equivalent to 
\begin{equation*}
 z_i(t) \leq  \alpha, ~~~~~~~~~\forall t \in [\hat{t},t^{k+1}_i]
\end{equation*}

As discussed in Section \ref{z_inside}, at instant $t^k_i$, agent $a_i$ does not know the future values of the neighbouring inputs in the interval $[t^k_i,t^{k+1}_i)$.
Thus, while maximizing the inter-update duration $\big(t^{k+1}_i - t^k_i\big)$, agent $a_i$ needs to consider the worst-case 
realizations of the neighbouring inputs $u_j,~\forall j\in S_i$, which result in the minimum value of the inter-update duration. 
This leads to the following max-min optimization:
\begin{align}
t^{k+1}_i = \arg~ \max\limits_{ \widetilde{t} \in {\mathbb{R}} }~&\min\limits_{ \substack{ u_i=\beta,\\ u_j \in \mathcal{U},\\ \forall j\in S_i } } ~~~ \widetilde{t} - t^k_i \label{max_min_time_out}\\
 & ~~~~\text{s.t.}~~~ z_i(t) \leq \alpha,~~~\forall t\in \big[t^k_i,\widetilde{t}\hspace*{0.1cm}\big] \label{constraint_out}
\end{align}
The following lemma shows that control law (\ref{update_instant_out})-(\ref{control_out}) is the 
solution of optimization (\ref{max_min_time_out})-(\ref{constraint_out}).
\begin{lemma}\label{opt_outside}
Consider any agent $a_i$ in MAS (\ref{MAS}).
 Let $z_i\big(t^k_i\big)$ be the local disagreement of $a_i$ at its update instant $t^k_i$ such that $z_i\big(t^k_i\big) < -\alpha$.
 Then, the control law (\ref{update_instant_out})-(\ref{control_out}) is
 the solution of optimization (\ref{max_min_time_out})-(\ref{constraint_out}).
\end{lemma}
\begin{proof}
See the Appendix for the proof.
\end{proof}

\begin{remark}\label{z_larger_1}
 The optimization (\ref{max_min_time_out})-(\ref{constraint_out}) and 
 Lemma \ref{opt_outside}
 correspond to the case $z_i\big(t^k_i\big) < -\alpha$. If  $z_i\big(t^k_i\big)> \alpha$, the constraint (\ref{constraint_out}) becomes
$z_i(t) \geq -\alpha ,~\forall t \in \big[t^k_i,\widetilde{t}\hspace*{0.08cm}\big]$.
Then, by following the arguments in the proof of Lemma \ref{opt_outside}, we can show that even for $z_i\big(t^k_i\big)> \alpha$, the control rule (\ref{update_instant_out})-(\ref{control_out})  
is the solution of optimization (\ref{max_min_time_out})-(\ref{constraint_out}). 
\end{remark}

\subsection{Feasibility of Protocol $\Romannum{1}$}\label{feasible1}
In this section, we present the proof of Lemma \ref{feasible} which claims that Protocol $\Romannum{1}$ satisfies the necessary condition (\ref{retain_inside}).
\begin{proof}\textit{of Lemma \ref{feasible}} : 
Consider MAS (\ref{MAS}) with any $n$, any connected communication graph $G$ and any initial condition $X(0) \in {\mathbb{R}}^n$.
Let $a_i$ be any agent in MAS (\ref{MAS}) and
$\widetilde{t}_i \in [t^k_i,t^{k+1}_i)$
be the first time instant under Protocol $\Romannum{1}$ at which $|z_i\big(\!\!~\widetilde{t}_i\big)| \leq \alpha$. If $|z_i\big(t^k_i\big)| \leq \alpha$, then it follows from (\ref{constraint}) and Lemma \ref{opt_inside} that
$|z_i(t)| \leq \alpha,~\forall t \in \big[~\!\widetilde{t}_i,t^{k+1}_i\big]$ under Protocol $\Romannum{1}$. On the other hand, if $|z_i\big(t^k_i\big)| > \alpha$, then it follows from  (\ref{constraint_out}), Lemma \ref{opt_outside} 
and Remark \ref{z_larger_1}
that $|z_i(t)| \leq \alpha,~\forall t \in \big[~\!\widetilde{t}_i,t^{k+1}_i\big]$ under Protocol $\Romannum{1}$. Consequently, in both cases, Lemma \ref{opt_inside} leads to $|z_i(t)| \leq \alpha,~\forall t > t^{k+1}_i$. This proves that Protocol $\Romannum{1}$ satisfies the necessary condition (\ref{retain_inside}).
\end{proof}
By using Lemma \ref{feasible}, it is already shown in Theorem \ref{consensus_time} that Protocol $\Romannum{1}$ satisfies constraint (\ref{constraint_main}) in
Problem \ref{prob_main} for time
$2T^{\ast}\big(X(0)\big)$, i.e., achieves $\alpha$ consensus of MAS (\ref{MAS})
within time $2T^{\ast}\big(X(0)\big)$, for every $n$, every connected $G$ and every 
$X(0) \in {\mathbb{R}}^n$. In the next section, we prove the optimality of Protocol $\Romannum{1}$.

\subsection{Proof of optimality of Protocol $\Romannum{1}$}\label{optimality}
Consider MAS (\ref{MAS}) with initial condition $X(0) \in {\mathbb{R}}^n$. Let $Q$ be any protocol which is a feasible solution of Problem \ref{prob_main} for  $T\big(X(0)\big)=2T^{\ast}\big(X(0)\big)$,
i.e., a
discrete communication based distributed protocol which achieves $\alpha$-consensus of MAS (\ref{MAS})
within time $2T^{\ast}\big(X(0)\big)$, for every $n$, every connected $G$ and every $X(0) \in {\mathbb{R}}^n$.
Let $C^Q_{MAS}\Big(2T^{\ast}\big(X(0)\big), X(0)\Big)$ and $C^{\ast}_{MAS}\Big(2T^{\ast}\big(X(0)\big),X(0)\Big)$ denote the value of the
aggregate communication cost $C_{MAS}$ defined in (\ref{agg_cost}), under Protocol $Q$ and Protocol $\Romannum{1}$, respectively.
\begin{theo}\label{proposed_optimal}
Consider MAS (\ref{MAS}) with initial condition $X(0) \in {\mathbb{R}}^n$. Then, under Protocol $\Romannum{1}$, the following holds.
\begin{equation}\label{C_less}
 C^{\ast}_{MAS}\Big(2T^{\ast}\big(X(0)\big), X(0)\Big) \leq C^Q_{MAS}\Big( 2T^{\ast}\big(X(0)\big), X(0)\Big)
\end{equation}
\end{theo}
\begin{proof}
 For the sake of contradiction, assume that
 \begin{equation}\label{less_Q}
  C^Q_{MAS}\Big(2T^{\ast}\big(X(0)\big), X(0)\Big) < C^{\ast}_{MAS}\Big( 2T^{\ast}\big(X(0)\big), X(0)\Big)
 \end{equation}
 Recall that $C_i\big(t,X(0)\big)$ denotes the number of update instants of agent $a_i$ in the interval $[0,t)$, corresponding to initial condition $X(0)$.
Let $C^Q_i \Big(2T^{\ast}\big(X(0)\big), X(0)\Big)$ and $C^{\ast}_i\Big(2T^{\ast}\big(X(0)\big),X(0)\Big)$ denote the value of
$C_i\Big(2T^{\ast}\big(X(0),X(0)\big) \Big)$, under Protocol $Q$ and Protocol $\Romannum{1}$, respectively.
Then, (\ref{agg_cost}) and (\ref{less_Q}) imply that there exists at least one agent, say $a_i$, such that
\begin{equation*}
 C^Q_i\Big(2T^{\ast}\big(X(0)\big), X(0)\Big) < C^{\ast}_i\Big( 2T^{\ast}\big(X(0)\big), X(0)\Big)
\end{equation*}
This implies that under Protocol $Q$, at least one inter-update duration  of agent $a_i$
is longer than that prescribed by control laws (\ref{update_instant}) and (\ref{update_instant_out})  in Protocol $\Romannum{1}$.

Recall from (\ref{max_min_time}) and Lemma \ref{opt_inside} that control law (\ref{update_instant}) gives the maximum inter-update duration under constraint (\ref{constraint}). Similarly, it follows from (\ref{max_min_time_out}) and Lemma \ref{opt_outside} that control law (\ref{update_instant_out}) gives the maximum inter-update duration under constraint (\ref{constraint_out}).
Then, as one inter-update duration under Protocol $Q$ is longer than that prescribed by control laws  (\ref{update_instant}) and (\ref{update_instant_out}), Protocol $Q$ must be violating either  constraint (\ref{constraint}) or constraint
(\ref{constraint_out}). 
 Recall that violation of (\ref{constraint}) or  (\ref{constraint_out}) by Protocol $Q$ results in the violation of  necessary 
 condition (\ref{retain_inside}) on feasible protocols. 
 This contradicts the fact that Protocol $Q$ is a feasible solution of Problem \ref{prob_main}
 and proves claim (\ref{C_less}).
\end{proof}

\section{Protocol $\Romannum{2}$: For $T\big(X(0)\big)>2T^{\ast}\big(X(0)\big)$}\label{T_larger}
In Section \ref{optimality}, we proved that Protocol $\Romannum{1}$ 
is the solution of Problem \ref{prob_main}
for $T\big(X(0)\big)=2T^{\ast}\big(X(0)\big)$. In this section, we extend Protocol $\Romannum{1}$ for 
$T\big(X(0)\big)>2T^{\ast}\big(X(0)\big)$. We refer to the extended protocol as Protocol $\Romannum{2}$.

\subsection{Protocol $\Romannum{2}$}\label{our_for_larger}
Let $\gamma > 1$ be a real number and $T\big(X(0)\big)=2\gamma T^{\ast}\big(X(0)\big)$ be the pre-specified $\alpha$-consensus time.
Define $\widetilde{\beta}:=\dfrac{\beta}{\gamma}$. Then, Protocol $\Romannum{2}$ is same as Protocol $\Romannum{1}$, except the 
control bound 
$\widetilde{\beta}$ in place of $\beta$.

\subsection{Optimality of Protocol $\Romannum{2}$}
In this section, we first show that Protocol $\Romannum{2}$ achieves $\alpha$-consensus of MAS (\ref{MAS}) within the pre-specified time $T\big(X(0)\big)=2\gamma T^{\ast}\big(X(0)\big)$. 
\begin{theo}\label{consensus_time_larger}
Consider MAS (\ref{MAS}). 
 Let $\alpha \in {\mathbb{R}}^{+}$
 be the pre-specified consensus bound.
 Then, for every $n$, every connected communication graph $G$ and every $X(0) \in {\mathbb{R}}^n$, Protocol $\Romannum{2}$
 achieves $\alpha$-consensus of MAS (\ref{MAS}) in time less than or equal to $2\gamma T^{\ast}\big(X(0)\big)$. 
 Moreover, there exist $\widetilde{n}$, connected graph $\widetilde{G}$
 and initial condition $\widetilde{X}(0) \in {\mathbb{R}}^{\widetilde{n}}$ 
 for which $\alpha$-consensus time under Protocol $\Romannum{2}$ is arbitrarily close to $2\gamma T^{\ast}\big(X(0)\big)$.
\end{theo}
\begin{proof}
 Recall that the control bounds in Protocols $\Romannum{1}$ and $\Romannum{2}$ are
 $\beta$ and $\widetilde{\beta}=\dfrac{\beta}{\gamma}$, respectively.
 Thus, the dynamics of MAS (\ref{MAS}) under Protocol $\Romannum{2}$ 
 is $\gamma$ times slower than that of under Protocol $\Romannum{1}$. Then, the claim follows from the arguments in the proofs of Theorems \ref{consensus_time} and \ref{strict_theo_1}.
\end{proof}

Now, we prove the optimality of Protocol $\Romannum{2}$  for $T\big(X(0)\big)=2\gamma T^{\ast}\big(X(0)\big)$.
Consider MAS (\ref{MAS}) with an initial condition $X(0) \in {\mathbb{R}}^n$. Let $Q$ be any protocol which is a feasible solution of Problem \ref{prob_main} for $T\big(X(0)\big)=2\gamma T^{\ast}\big(X(0)\big)$.
Let $C^Q_{MAS}\Big(2\gamma T^{\ast}\big(X(0)\big), X(0)\Big)$ and $C^{\ast}_{MAS}\Big(2\gamma T^{\ast}\big(X(0)\big),X(0)\Big)$ denote the value of $C_{MAS}$ defined in (\ref{agg_cost}), under Protocol $Q$ and Protocol $\Romannum{2}$, respectively.
\begin{theo}\label{proposed_optimal_larger}
Consider MAS (\ref{MAS}) with initial condition $X(0) \in {\mathbb{R}}^n$. Then, under Protocol $\Romannum{2}$, the following holds.
\begin{equation*}
 C^{\ast}_{MAS}\Big(2\gamma T^{\ast}\big(X(0)\big), X(0)\Big) \leq C^Q_{MAS}\Big( 2\gamma T^{\ast}\big(X(0)\big), X(0)\Big)
\end{equation*}
\end{theo}
\begin{proof}
 The claim follows from the arguments in the proof of Theorem \ref{proposed_optimal}.
\end{proof}

Next, we analyze the effect of $T\big(X(0)\big)$ on $C^{\ast}_{MAS}$. Intuitively, with the increase in $T\big(X(0)\big)$, agents in MAS (\ref{MAS}) can afford to delay their next update instants, which would result in lower $C^{\ast}_{MAS}$. However, the following theorem shows that the above intuition is not true after $T\big(X(0)\big)=2 T^{\ast}\big(X(0)\big)$. 
\begin{theo}\label{C_MAS_constant}
Consider MAS (\ref{MAS}) with initial condition $X(0) \in {\mathbb{R}}^n$. Then, under Protocol $\Romannum{2}$, for every real $\gamma > 1$,
the following holds.
\begin{equation}\label{C_MAS_constant_2T}
 C^{\ast}_{MAS}\Big(2\gamma T^{\ast}\big(X(0)\big), X(0)\Big) = C^{\ast}_{MAS}\Big( 2 T^{\ast}\big(X(0)\big), X(0)\Big)
\end{equation}
\end{theo}
\begin{proof}
Let $t^{k,\!~\beta}_i$ and $t^{k,\!~\widetilde{\beta}}_i$ denote the $k$th update instants of agent $a_i$
 under Protocols $\Romannum{1}$ and $\Romannum{2}$, respectively.
 Recall that the dynamics of MAS (\ref{MAS}) under Protocol $\Romannum{2}$ 
 is $\gamma$ times slower than that of under Protocol $\Romannum{1}$.
Then, it is easy to show that
 \begin{equation*}
 t^{k,\!~\widetilde{\beta}}_i = \gamma ~\!t^{k,\!~\beta}_i,~~~~~~~~\forall i \in P,~~~~~~~~~\forall k \geq 1
 \end{equation*}
 Hence, the timeline of agent $a_i$ under Protocol $\Romannum{2}$ is just the $\gamma$-scaled version of its timeline 
 under Protocol $\Romannum{1}$.
 As a result, the number of update instants of agent $a_i$ in the interval $\big[0,\!~2\gamma T^{\ast}\big(X(0)\big)\big]$ under Protocol $\Romannum{2}$ is equal to that of in the interval $\big[0,\!~2 T^{\ast}\big(X(0)\big)\big]$ under Protocol $\Romannum{1}$. Then, the claim (\ref{C_MAS_constant_2T}) follows from the definition of $C_{MAS}$ in (\ref{agg_cost}).
\end{proof}

\section{Simulation results}\label{simulation}
In this section, we present the simulation results obtained under Protocol $\Romannum{2}$ for various $T\big(X(0)\big)$'s.
Consider MAS (\ref{MAS}) with $n=6$ agents and the communication graph shown in Fig. \ref{Simulation_graph}.
 \begin{figure}[h]
\centering
\begin{tikzpicture}[scale=0.39]

\draw (3,1) circle (1cm);

\draw (3,-4.5) circle (1cm);

\draw (7,-1.75) circle (1cm);

\draw (13,-1.75) circle (1cm);

 \draw (17,1) circle (1cm);
 
 \draw (17,-4.5) circle (1cm);
 
 \draw [] (3.8,0.44) to (6.2,-1.2);
 
 \draw [] (3.8,-3.94) to (6.2,-2.3);
 
 \draw [] (8,-1.75) to (12,-1.75);
 
 \draw [] (16.2,0.44) to (13.8,-1.2);
 
  \draw [] (16.2,-3.94) to (13.8,-2.3);
  
  \node at (3,1) {$1$};
  \node at (3,-4.5) {$2$};
  \node at (7,-1.75) {$3$};
  \node at (13,-1.75) {$4$};
  \node at (17,1) {$5$};
  
  \node at (17,-4.5) {$6$};

\end{tikzpicture}
\caption{Communication graph for simulation}
\label{Simulation_graph}
\end{figure}
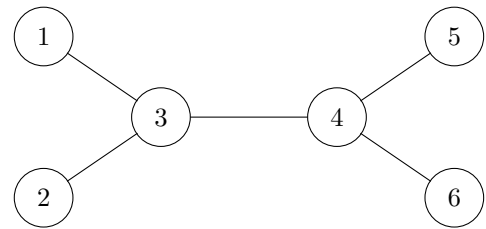
Let the consensus bound and the control bound  be $\alpha=0.6$ and $\beta=1$, respectively. Let the initial condition of the MAS be $X(0)=[7,2,4,3,1,5]$. Then, it follows from the definition of $T^{\ast}\big(X(0)\big)$ in (\ref{min_time_expression}) that $T^{\ast}\big(X(0)\big)=3$ sec.

We simulated the above MAS under Protocol $\Romannum{2}$ for $T\big(X(0)\big)= 2T^{\ast}\big(X(0)\big)$, $ T\big(X(0)\big)=10T^{\ast}\big(X(0)\big)$ and $T\big(X(0)\big)= 20T^{\ast}\big(X(0)\big)$. The corresponding disagreement trajectories $z_i$'s are plotted in Fig. \ref{Simulation_equal_time}, \ref{Simulation_multiplier_5} and \ref{Simulation_multiplier_10}, respectively. The corresponding $\alpha$-consensus times and optimal communication costs are given in Table $\Romannum{1}$. 

Notice that in each of the above three cases, Protocol $\Romannum{2}$ achieves $\alpha$-consensus within the prespecified time $T\big(X(0)\big)$, as proved in Theorem \ref{consensus_time_larger}. Further, the $\alpha$-consensus times corresponding to $T\big(X(0)\big)= 10T^{\ast}\big(X(0)\big)$ and $T\big(X(0)\big)= 20T^{\ast}\big(X(0)\big)$ are $5$ and $10$ times of that corresponding to $T\big(X(0)\big)= 2T^{\ast}\big(X(0)\big)$, respectively. This observation is in accordance with the arguments in the proof of Theorem \ref{consensus_time_larger}. In addition, observe that the trajectories in Figures \ref{Simulation_multiplier_5} and \ref{Simulation_multiplier_10} are just the time-stretched versions of the corresponding trajectories in Fig. \ref{Simulation_equal_time}. The optimal communication costs $C^{\ast}_i$'s and $C^{\ast}_{MAS}$ corresponding to
 $T\big(X(0)\big)= 2T^{\ast}\big(X(0)\big)$,
 $T\big(X(0)\big)= 10T^{\ast}\big(X(0)\big)$ and  $T\big(X(0)\big)= 20T^{\ast}\big(X(0)\big)$ are same, as proved in Theorem \ref{C_MAS_constant}.

\begin{figure}
     \centering
     \begin{subfigure}[t]{1\columnwidth}
         \centering
         \includegraphics[width=\textwidth]{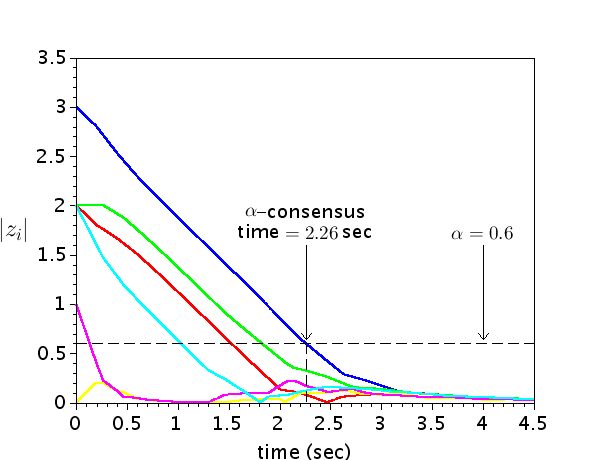}
         \caption{$T\big(X(0)\big) = 2T^{\ast}\big(X(0)\big)$}
         \label{Simulation_equal_time}
         \qquad
     \end{subfigure}
     \begin{subfigure}[t]{1\columnwidth}
         \centering
         \includegraphics[width=\textwidth]{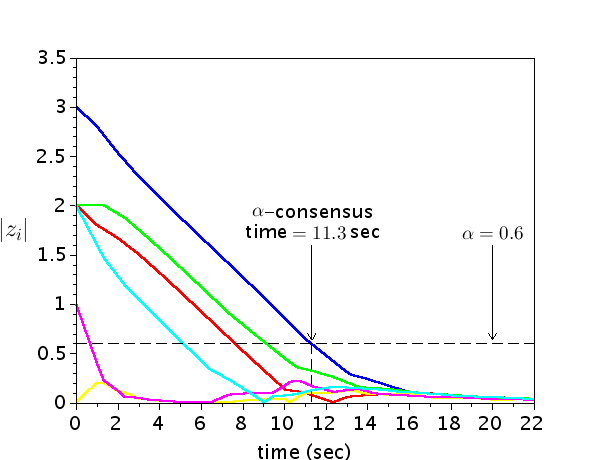}
         \caption{$T\big(X(0)\big) = 10T^{\ast}\big(X(0)\big)$}
         \label{Simulation_multiplier_5}
         \qquad
     \end{subfigure}
     \begin{subfigure}[t]{1\columnwidth}
         \centering
         \includegraphics[width=\textwidth]{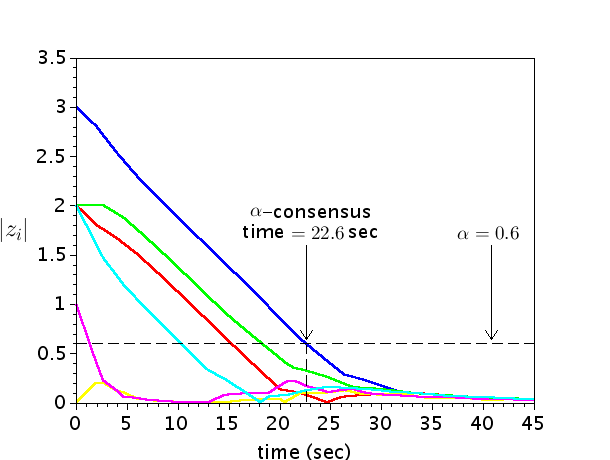}
         \caption{$T\big(X(0)\big) = 20T^{\ast}\big(X(0)\big)$}
         \label{Simulation_multiplier_10}
         \qquad
     \end{subfigure}
     \caption{Comparison of local disagreement trajectories under Protocol $\Romannum{2}$ for various $T\big(X(0)\big)$'s}
     \label{Simulation_disgreement}
\end{figure}

\begin{table*}
 \caption{Consensus time and communication cost under Protocol $\Romannum{2}$ for various $T\big(X(0)\big)$'s}
\label{simulation_table}
\begin{tabularx}{\textwidth}{@{}l*{10}{C}c@{}}
\toprule
$T\big(X(0)\big)$ & $\alpha$-consensus time [sec] & $C^{\ast}_1$ &
$C^{\ast}_2$ &
$C^{\ast}_3$ &
$C^{\ast}_4$ &
$C^{\ast}_5$ &
$C^{\ast}_6$ & 
$C^{\ast}_{MAS}$ \\
\midrule
\midrule
$2T^{\ast}\big(X(0)\big)$ & 
 $2.26$ & $8$ & $9$ & $30$ & $30$ & $9$ & $9$ & $95$ \\
 \midrule
$10T^{\ast}\big(X(0)\big)$ & 
 $11.3$ & $8$ & $9$ & $30$ & $30$ & $9$ & $9$ & $95$\\ 
 \midrule
 $20T^{\ast}\big(X(0)\big)$ & 
 $22.6$ & $8$ & $9$ & $30$ & $30$ & $9$ & $9$ & $95$\\
\bottomrule
\end{tabularx}
\end{table*}

\section{Conclusion and future work}\label{conclusion1}
In this paper, a distributed consensus protocol (Protocol $\Romannum{2}$) is proposed for a MAS of single-integrator agents with bounded inputs and time-invariant communication graph. 
We showed (Theorems \ref{consensus_time_larger} and \ref{proposed_optimal_larger}) that the proposed protocol minimizes the aggregate number of communication instants required to achieve $\alpha$-consensus within the pre-specified time $T\big(X(0)\big)\geq 2T^{\ast}\big(X(0)\big)$. The control rules in the proposed protocol are obtained by maximizing the inter-update durations. We computed (Lemmas \ref{opt_inside} and \ref{opt_outside}) the closed form solutions of these optimizations. Finally, we presented the simulation results which verify the theoretical claims. The work of extending the proposed protocol to MASs with complex dynamics is in progress.

\section*{Appendix}

\subsection*{A.~Intermediate results for the proof of Theorem \ref{consensus_time} }
In this section, we develop a few intermediate results which will be used later in the proof of Theorem \ref{consensus_time}.
Let $x_i\big(t^k_i\big)$ and $z_i\big(t^k_i\big)$ be the state and the local disagreement of agent $a_i$ respectively, at its update instant $t^k_i$. 
Recall that $S_i$ is the set of indices of neighbours of agent $a_i$, with cardinality $n_i$.
\begin{lemma}\label{average_inside}
Consider agent $a_i$ in MAS  (\ref{MAS}) with $|z_i\big(t^k_i\big)|\leq \alpha$. Then, under Protocol $\Romannum{1}$, the following holds.\\
$1$.~$x_i\big(t^{k+1}_i\big) = \dfrac{ \sum_{ j\in S_i} x_j\big(t^k_i\big) }{n_i}$\\
$2$.~$|z_i(t)| \leq \alpha,\hspace*{0.351cm}\forall t\geq t^k_i$
\end{lemma}
\begin{proof}

At instant $t>t^k_i$, the state $x_i$ of agent $a_i$ is $x_i(t) = x_i(t^k_i) + \int_{t^k_i}^{t} u_i(\tau) d \tau$.
It is given that $|z_i(t^k_i)| \leq \alpha$. Then, control rule (\ref{update_instant})-(\ref{control}) in Protocol $\Romannum{1}$ lead to
\begin{equation*}
 x_i\big(t^{k+1}_i\big)=x_i\big(t^k_i\big) + \!\! \int_{t^k_i}^{t^{k}_i + \dfrac{\alpha}{\beta n_i} } \frac{-z_i(t^k_i)}{\alpha}\beta d\tau = x_i\big(t^k_i\big) -\dfrac{ z_i\big(t^k_i\big)}{n_i}
 \end{equation*}
Then, it follows from the definition of $z_i$ in (\ref{local_dis}) that
\begin{equation}\label{x_average_2}
x_i\big(t^{k+1}_i\big) = x_i\big(t^k_i\big) - \dfrac{ \sum_{j\in S_i}\Big(x_i\big(t^k_i\big)-x_j\big(t^k_i\big)\Big) }{n_i} 
\end{equation}
Recall that the cardinality of $S_i$ is $n_i$. Then, by re-arranging the terms in (\ref{x_average_2}), we get
$x_i\big(t^{k+1}_i\big) = \dfrac{ \sum_{ j\in S_i} x_j\big(t^k_i\big) }{ n_i}$.
 This proves the claim in Lemma \ref{average_inside}.$1$.

Next, we prove that under Protocol $\Romannum{1}$, $z_i$ satisfies $|z_i(t)| \leq \alpha,~\forall t\in [t^k_i,t^{k+1}_i]$.
Then, the claim in Lemma \ref{average_inside}.$2$ follows by repeating arguments in the subsequent inter-update intervals.

Let the control input $u^{\ast}_i$ be as defined in (\ref{control}). We know from the dynamics of $x_i$ in (\ref{MAS}) and  the definition of $z_i$ in (\ref{local_dis}) that
 \begin{equation}\label{z_in_proof_1_in}
  z_i\big(t^{k+1}_i\big)=z_i\big(t^k_i\big) + n_i \int_{t^k_i}^{t^{k+1}_i}\!\!\!\!\!\!\!\!u^{\ast}_i(t)~ dt - \sum_{j\in S_i} \Bigg ( \int_{t^k_i}^{t^{k+1}_i}\!\!\!\!\!\!\!\!u_j(t)~ dt \Bigg )
 \end{equation}
By using (\ref{update_instant}) and (\ref{control}), it is easy to show that
 \begin{equation}\label{z_in_proof_2}
  z_i\big(t^k_i\big) + n_i \int_{t^k_i}^{t^{k+1}_i}\!\!\!\!\!\!\!\!u^{\ast}_i(t)~ dt = 0
 \end{equation}
Further, recall that $|u_j(t)| \leq \beta,~\forall t\geq 0,~\forall j\in P$. Then, it follows from (\ref{update_instant}) that
\begin{equation}\label{z_in_proof_3}
 \Bigg|\int_{t^k_i}^{t^{k+1}_i}\!\!\!\!\!\!\!\!u_j(t)~dt\Bigg| = \Bigg|\int_{t^k_i}^{t^k_i + \dfrac{\alpha}{\beta n_i} }u_j(t)~dt\Bigg| \leq \frac{\alpha}{n_i} 
\end{equation}
Then, (\ref{z_in_proof_1_in}), (\ref{z_in_proof_2}) and (\ref{z_in_proof_3}) together lead to
\begin{equation}\label{z_at_t_k_1}
\big|z_i\big(t^{k+1}_i\big)\big|\leq \alpha
\end{equation}

Now, we claim that $|z_i(t)|\leq \alpha,~\forall t\in [t^k_i,t^{k+1}_i]$.
Assume for the sake of contradiction that there exists $\hat{t}\in (t^k_i,t^{k+1}_i)$ such that
$|z_i\big(\hat{t}\hspace*{0.04cm}\big)|> \alpha$. Take $u_j(t)=u^{\ast}_i(t),~\forall t \in \big[\hat{t},t^{k+1}_i\big),~\forall j\in S_i$.
Then, it follows from (\ref{z_in_proof_1_in}) that $|z_i\big(t^{k+1}_i\big)|=|z_i\big(\hat{t}\hspace*{0.04cm}\big)| > \alpha$, which contradicts (\ref{z_at_t_k_1})
and proves our claim that
 $|z_i(t)|\leq \alpha, ~\forall t\in [t^k_i,t^{k+1}_i]$.
 By repeating the above arguments in the subsequent inter-update intervals of agent $a_i$, we get Lemma \ref{average_inside}.$2$.
This completes the proof.
\end{proof}

Let $t^k_i$ be an update instant of agent $a_i$ in MAS (\ref{MAS}).
Define 
\begin{align}
x_i^{min}\big(t^k_i\big)&:= \min_{j\in (S_i \cup \hspace*{0.03cm}i) } ~x_j\big(t^k_i\big) \label{x_i_min_def}\\ x_i^{max}\big(t^k_i\big)&:= \max_{j\in (S_i \cup \hspace*{0.03cm}i)} ~x_j\big(t^k_i\big) \label{x_i_max_def}
\end{align}
The following lemma shows that the state $x_i$ remains bounded between $x_i^{min}\big(t^k_i\big)$ and $x_i^{max}\big(t^k_i\big)$, in the interval $[t^k_i,t^{k+1}_i]$.

\begin{lemma}\label{inside_one_interval}
Consider any agent $a_i$ in MAS (\ref{MAS}). Under Protocol $\Romannum{1}$, the state of agent $a_i$ satisfies the following.
 \begin{equation}\label{x_i_t_between_min_max}
 x_i^{min}\big(t^k_i\big) \leq x_i(t) \leq x_i^{max}\big(t^k_i\big),~~~~~\forall t\in \big[t^k_i,t^{k+1}_i\big]
\end{equation}
\end{lemma}
\begin{proof}
 Based on the value of $z_i\big(t^k_i\big)$, there are two cases, $|z_i\big(t^k_i\big)| \leq \alpha$ and $|z_i\big(t^k_i\big)| 
 >\alpha$. We will prove the claim for the first case. The proof for the second case follows analogously.

Without loss of generality, assume that $-\alpha \leq z_i\big(t^k_i\big) \leq 0$.
Then, it follows from (\ref{control}) that $u_i(t) = \dfrac{-z_i\big(t^k_i\big)}{\alpha}\beta \geq 0,~\forall t \in [t^k_i,t^{k+1}_i)$.
Hence, $\dot{x}_i(t)=u_i(t) \geq 0,~\forall t \in [t^k_i,t^{k+1}_i)$, which implies that
 $x_i\big(t^k_i\big)  \leq x_i(t),~\forall t \in [t^k_i,t^{k+1}_i]$.
 Then, the definition of $x_i^{min}\big(t^k_i\big)$ in (\ref{x_i_min_def}) leads to
\begin{equation}\label{x_larger_inside}
 x_i^{min}\big(t^k_i\big) \leq x_i\big(t^k_i\big) \leq x_i(t),~~~~~~\forall t \in \big[t^k_i,t^{k+1}_i\big]
\end{equation}

Next, let $q \in (S_i \cup i)$ be such that $x_q\big(t^k_i\big)=x_i^{max}\big(t^k_i\big)$.
Then, based on the value of $q$, there are following two cases:\\
\text{Case 1 : } $q=i$\\
In this case, we must have $x_i\big(t^k_i\big) = x_j\big(t^k_i\big),~\forall j\in S_i$. Otherwise, we will have $z_i\big(t^k_i\big)>0$, which violates the assumption $-\alpha \leq z_i\big(t^k_i\big) \leq 0$.
Thus, $z_i\big(t^k_i\big)=0$. Then, (\ref{control}) leads to $u_i(t)=0,~\forall t \in [t^k_i,t^{k+1}_i)$ and
hence, $x_i\big(t^{k+1}_i\big)=x_i\big(t^k_i\big) = x_i^{max}\big(t^k_i\big)$.\\
\text{Case 2 : } $q \in S_i$\\
Recall from Lemma \ref{average_inside}.$1$ that
$x_i\big(t^{k+1}_i\big)$ is the average of 
$x_j\big(t^k_i\big)$'s for all $j \in S_i$.
Thus, $x_i\big(t^{k+1}_i\big) \leq \max_{j \in S_i} x_j\big(t^k_i\big) = x_i^{max}\big(t^k_i\big)$. 

In both of the above cases, we have $x_i\big(t^{k+1}_i\big) \leq x_i^{max}\big(t^k_i\big)$.
Recall that $\dot{x}_i(t) \geq 0,~\forall t \in [t^k_i, t^{k+1}]$. Thus,
\begin{equation}\label{x_smaller_inside}
 x_i(t) \leq x_i\big(t^{k+1}_i\big) \leq x^{max}_i\big(t^k_i\big),~~~~~\forall t\in \big[t^k_i,t^{k+1}_i\big]
\end{equation}
Then, the claim (\ref{x_i_t_between_min_max}) follows from (\ref{x_larger_inside}) and (\ref{x_smaller_inside}).
\end{proof}

Let $x^{min}(0)$ and $x^{max}(0)$ be as in the definition of $T^{\ast}\big(X(0)\big)$ in (\ref{min_time_expression}). The following lemma shows that the states of all agents in MAS (\ref{MAS}) remain bounded between $x^{min}(0)$ and $x^{max}(0)$, for all $t\geq 0$.
\begin{lemma}\label{inside_initial_val}
Consider MAS (\ref{MAS}) with initial condition $X(0)\in {\mathbb{R}}^n$.
Under 
Protocol $\Romannum{1}$, the following holds.
\begin{equation}\label{x_i_in_min_max_t_each}
 x^{min}(0) \leq x_i(t) \leq x^{max}(0),~~~~~\forall t\geq 0,~~~~~ \forall i \in P
\end{equation}
\end{lemma}
\begin{proof}
 Recall that according to Protocol $\Romannum{1}$, the time instant $t^1_i=0$ is the first update instant of every agent $a_i,~i \in P$.
 Let $x_i^{min}\big(t^k_i\big)$ and $x_i^{max}\big(t^k_i\big)$ be as defined in (\ref{x_i_min_def}) and (\ref{x_i_max_def}), respectively.
 Then, it follows from Lemma \ref{inside_one_interval} that 
 \begin{equation}\label{x_between_min_and_max}
  x_i^{min}(0) \leq x_i(t) \leq x_i^{max}(0),~~~~\forall t \in [0, t^2_i],~~~\forall i \in P 
 \end{equation}
Further, it is clear from the definitions of $x_i^{min}$ and $x^{min}(0)$ that $x^{min}(0) \leq x_i^{min}(0),~\forall i \in P$.
Similarly,  $x_i^{max}(0) \leq x^{max}(0),~\forall i \in P$. Then, (\ref{x_between_min_and_max}) leads to
\begin{equation*}
  x^{min}(0) \leq x_i(t) \leq x^{max}(0),~~~~\forall t \in [0, t^2_i],\hspace*{0.5cm} \forall i\in P 
\end{equation*}
Now, by applying the above arguments in the subsequent inter-update intervals $[t^2_i,t^3_i)$, $[t^3_i,t^4_i),\dots$ of all agents, we get (\ref{x_i_in_min_max_t_each}).
This completes the proof.
\end{proof}

\subsection*{B.~Proof of Theorem \ref{consensus_time}} 
For the sake of contradiction, assume that the claim is not correct, i.e., for some $n$, $G$ and $X(0)$, Protocol $\Romannum{1}$
 does not achieve $\alpha$-consensus of MAS (\ref{MAS}) within time $2T^{\ast}\big(X(0)\big)$. 
 This implies that there exists at least one agent, say $a_i$, and a time instant $\hat{t}_i > 2T^{\ast}\big(X(0)\big)$,
 such that $|z_i(\hat{t}_i)| > \alpha$.
 Recall from Lemma \ref{feasible} that Protocol $\Romannum{1}$
 satisfies the
 necessary condition (\ref{retain_inside}).
 Then, it follows from (\ref{retain_inside}) that under Protocol $\Romannum{1}$, the following holds.
 \begin{equation*}
  z_i(t) < -\alpha,~\forall t \in [0, \hat{t}_i] \hspace*{0.6cm} or \hspace*{0.6cm} z_i(t) > \alpha,~\forall t \in [0, \hat{t}_i]
 \end{equation*}
 Without loss of generality, assume that $z_i(t) < -\alpha,~\forall t \in [0, \hat{t}_i]$.
Then, it follows from control rule (\ref{control_out}) that $u_i(t)=\beta,~\forall t\in [0,\hat{t}_i]$ and hence, 
$x_i\big(\hat{t}_i\big) = x_i(0) + \beta \hat{t}_i$.           
As $\hat{t}_i > 2T^{\ast}\big(X(0)\big)$, we get
\begin{equation}\label{x_i_larger}
x_i\big(\hat{t}_i\big) > \Big(x_i(0) + 2\beta T^{\ast}\big(X(0)\big)\Big)
\end{equation}
Let $x^{min}(0)$ and $x^{max}(0)$ be as in the definition of $T^{\ast}\big(X(0)\big)$ in (\ref{min_time_expression}).
Then, (\ref{min_time_expression}) and (\ref{x_i_larger}) lead to
\begin{equation}\label{x_greater}
 x_i\big(\hat{t}_i\big) > \Big(x_i(0) + x^{max}(0) - x^{min}(0)\Big)
\end{equation}
It is clear from the definition of $x^{min}(0)$ that
$x_i(0)\geq x^{min}(0)$.
Then, (\ref{x_greater}) leads to 
$x_i(\hat{t}_i) > x^{max}(0)$.
This contradicts Lemma \ref{inside_initial_val} and proves the claim in Theorem \ref{consensus_time}.

\subsection*{C.~Intermediate result for the proof of Theorem \ref{strict_theo_1}}
\begin{lemma}\label{x_i_close_5}
Consider the MAS in Example \ref{strict_bound} with the communication graph $G_e$ shown in Fig. \ref{graph_in_example}.
Let $r$ be the number of neighbours of agent $a_1$.
Then, for every $\mu \in {\mathbb{R}}^+$, 
there exists an integer $n=n_{\mu,r}$ such that under Protocol $\Romannum{1}$, the following holds. 
\begin{equation}\label{claim}
 5-\mu \leq x_i(t) \leq 5,~~~~~~~\forall t\geq \mu,~~~~~~~~\forall i=2,3,\dots,n_{\mu,r}
\end{equation}
\end{lemma}
\begin{proof}
We first define the integer $n_{\mu,r}$ for which the claim (\ref{claim}) holds.
Let $\mu \in {\mathbb{R}}^+$ be the given number. Define $\widetilde{\mu}:=\dfrac{\mu}{2}$. 
Let \textit{ceil} be the standard ceiling function. Define the integers 
$n_{{\mu}_1}:=1+ ceil\bigg(\dfrac{4}{\widetilde{\mu}}\bigg)$,
$n_{{\mu}_2}:= ceil\bigg(\dfrac{5}{\mu-\widetilde{\mu}}\bigg)$,
$n_{r_1}:=1+\bigg(\dfrac{16r}{5r+3}\bigg)$ and
$n_{r_2}:=3r-7$.
Then, define
\begin{equation}\label{define_n_u_r}
n_{{\mu},r}:=\max\big\{8, n_{{\mu}_1}, n_{{\mu}_2}, n_{r_1}, n_{r_2}\big\}  + 1
\end{equation}
Now, consider the MAS in Example \ref{strict_bound} with $n=n_{{\mu},r}$.
 We will show that the claim (\ref{claim}) holds for this MAS.
For the sake of clarity, we divide the remaining proof into five parts as follows.\\
\textit{$1)$ Computation of $t^2_i$ and $x_i\big(t^2_i\big)$ for $i\geq 1$}

According to Protocol $\Romannum{1}$, the first update instant of every agent is $t^1_i=0$. At this instant, the states of agents are $x_1(0)=0$ and $x_i(0)=5,~\forall i=2,\dots,n_{{\mu},r}$. Then,
the local disagreement of agent $a_1$ at instant $t=0$ is $z_1(0)=-5r$. Note that $r\geq 1$. Thus, $z_1(0) <-3=-\alpha$.
Then, it follows from control rule (\ref{update_instant_out})-(\ref{control_out}) that
 \begin{align}
 t^{2}_1=~\!\!&\dfrac{|z_1(0)|+\alpha}{2\beta n_1} = \dfrac{5r+3}{2r} \label{agent_1_next} \\ u^{\ast}_1(t)=~\!\!&\beta=1,~~~~~~~\forall t \in \big[0,t^2_1\big) \label{agent_1_control}
 \end{align}
 As a result,
$x_1\big(t^2_1\big) = x_1(0) + t^2_1 =  \dfrac{5r+3}{2r}$.

Now, consider any agent $a_l,~l=2,\dots,r+1$. Note that $n_l=n_{{\mu},r}-1$. Thus, the local disagreement of agent $a_l$ at instant $t=0$ is
$z_l(0)=(5-0) + (n_{{\mu},r}-2)(5-5) = 5>\alpha$.
Then, it follows from control rule  (\ref{update_instant_out})-(\ref{control_out}) that
\begin{align}
 t^{2}_l=&\dfrac{|z_l(0)|+\alpha}{2\beta n_l}=\dfrac{4}{n_{{\mu},r}-1} \label{agent_2_1_next}\\
 u^{\ast}_l(t)=&-\beta=-1,~~~~~\forall t \in \big[0,t^2_l\big)\label{agent_2_1_control}
\end{align}
As a result, $x_l\big(t^2_l\big) = x_l(0) -t^2_l = 5-\dfrac{4}{n_{{\mu},r}-1}$.
It follows from the definitions of $n_{{\mu}_1}$ and $n_{{\mu},r}$ in (\ref{define_n_u_r}) that $\dfrac{4}{n_{{\mu},r}-1} < \widetilde{\mu}$. Thus,  $x_l\big(t^2_l\big)$ satisfies
$(5- \widetilde{\mu})< x_l\big(t^2_l\big)<5$.
Recall from (\ref{agent_2_1_control}) that $\dot{x}_l(t) < 0,~\forall t \in [0,t^2_l)$. Recall also that $x_l(0)=5$.
Hence,
\begin{equation}\label{x_l_less_u}
 (5- \widetilde{\mu})< x_l(t)<5,~~ \forall t \in \big[0,t^2_l\big],~~ \forall l=2,\dots,r+1
 \end{equation}

Next, consider any agent $a_j,~j=r+2,\dots,n_{{\mu},r}$. Note that $n_j=n_{{\mu},r}-2$.
Thus, the local disagreement of agent $a_j$  at instant $t=0$ is $z_j(0)= (5-5)(n_{{\mu},r}-2) = 0$.
Then, it follows from control rule (\ref{update_instant})-(\ref{control}) that
 \begin{align}
 t^{2}_j = ~\!\! &\dfrac{\alpha}{\beta n_j} = \dfrac{3}{n_{{\mu},r}-2} \label{agent_r2_next} \\ u^{\ast}_j(t) = ~\!\!&0,~~~~~~\forall t \in \big[0,t^2_j\big)  \label{agent_r2_control}
 \end{align}
 As a result, the state $x_j$ satisfies
 \begin{equation}\label{x_j_second}
 x_j(t) = x_j(0) =5,~~~~\forall t \in \big[0,t^2_j\big],~~~~\forall j=r+2,\dots,n_{{\mu},r}
 \end{equation}

\noindent \textit{$2$) Ordering of update instants $t^2_i$'s}

Let $t^2_i$'s be as given in (\ref{agent_2_1_next}) and (\ref{agent_r2_next}).
Recall from (\ref{define_n_u_r}) that $n_{{\mu},r} > 8$.
Then, it is easy to show that $\dfrac{3}{n_{{\mu},r}-2} < \dfrac{4}{n_{{\mu},r}-1}$. Then, it follows from (\ref{agent_2_1_next}) and (\ref{agent_r2_next}) that 
\begin{equation}\label{t_2_order}
    t^2_{r+2}=t^2_{r+3}=\dots=t^2_{n_{{\mu},r}} < t^2_2=t^2_3=\dots=t^2_{r+1}
\end{equation}

Let $t^2_1$ be as given in (\ref{agent_1_next}).
It follows from the definitions of $n_{r_1}$ and $n_{{\mu},r}$ in (\ref{define_n_u_r})
that
$\dfrac{4}{n_{{\mu},r}-1} < \dfrac{5r+3}{2r}$. Then, (\ref{agent_1_next}) and (\ref{agent_2_1_next})  lead to
\begin{equation}\label{t_2_1_2}
t^2_2=t^2_3=\dots=t^2_{r+1} < t^2_1
\end{equation}

\noindent \textit{$3$) Computation of $t^3_i$ for $i \geq 2$}

Consider any agent $a_l,~l=2,\dots,r+1$.
 Recall from (\ref{t_2_order}) and (\ref{t_2_1_2}) that $t^2_{r+2} < t^2_l  < t^2_1$.
 Now, we compute $z_l\big(t^2_{r+2}\big)$. 
 According to (\ref{agent_1_control}) and (\ref{agent_r2_next}), the state of agent $a_1$ at instant $t^2_{r+2}$ is
 $x_1\big(t^2_{r+2}\big) = \dfrac{3}{n_{{\mu},r}-2}$.
 Similarly, according to (\ref{agent_2_1_control}) and (\ref{agent_r2_next}), the state of agent $a_l,~l=2,\dots,r+1$, at instant $t^2_{r+2}$ is $x_l\big(t^2_{r+2}\big) = 5-\dfrac{3}{n_{{\mu},r}-2}$.  Recall from (\ref{x_j_second}) that $x_j\big(t^2_{r+2}\big) = 5,~\forall j=r+2,\dots,n_{{\mu},r}$. Thus, for $l=2,\dots,r+1$, we get
 \begin{equation}\label{z_t_r_2}
 z_l\big(t^2_{r+2}\big) = \!\!\!
 \sum_{k=1, k \neq l}^{n_{{\mu},r}} \!\!\!\Big(x_l\big(t^2_{r+2}\big) -  x_k\big(t^2_{r+2}\big)\Big)
\! =5-\dfrac{3(n_{{\mu},r}-r+1)}{n_{{\mu},r}-2}
 \end{equation}
Then, it follows from definitions of $n_{r_2}$ and $n_{{\mu},r}$ in (\ref{define_n_u_r}) that 
 $z_l\big(t^2_{r+2}\big) \in \![2,3] \!\subset \![-\alpha,\alpha]$. \!After that,
Lemmas \ref{necessary_condition} and \ref{feasible} give \begin{equation}\label{z_2_2_2}
 |z_l(t)| \leq \alpha,~~\forall t \geq t^2_{r+2},~~ \forall l=2,\dots,r+1
 \end{equation}
 Recall from (\ref{t_2_order}) that $t^2_{r+2} <t^2_l, \forall l=2,\dots,r+1$. Hence,
\begin{equation}\label{z_l_2_less}
 |z_l\big(t^2_l\big)| \leq \alpha,~~~ \forall l=2,\dots,r+1
 \end{equation}
 Then, (\ref{update_instant}) and (\ref{agent_2_1_next}) together lead to
 \begin{equation}\label{t_3_2}
  t^3_l = \bigg(t^2_l + \dfrac{\alpha}{\beta n_l}\bigg) =  \dfrac{7}{n_{{\mu},r}-1},~\forall l=2,\dots,r+1
  \end{equation}
 
Next, consider any agent $a_j,~j=r+2,\dots,n_{{\mu},r}$.
Recall that $z_j(0)=0$. Then, it follows from Lemma \ref{average_inside}.$2$ that
\begin{equation}\label{z_r2}
|z_j(t)| \leq \alpha,~~\forall t\geq 0,~~\forall j=r+2,\dots,n_{{\mu},r}
\end{equation}
 Subsequently, (\ref{update_instant}) and (\ref{agent_r2_next}) together lead to
 \begin{equation}\label{t_3_r2}
 t^{3}_j=t^2_j+\dfrac{\alpha}{\beta n_j}=\dfrac{6}{n_{{\mu},r}-2},~~~~\forall j=r+2,\dots,n_{{\mu},r}
 \end{equation}
 Recall from (\ref{define_n_u_r}) that $n_{{\mu},r} > 8$. Then, it is easy to show that
 $\dfrac{4}{n_{{\mu},r}-1}<\dfrac{6}{n_{{\mu},r}-2}< \dfrac{7}{n_{{\mu},r}-1}$. Then, (\ref{agent_2_1_next}), (\ref{t_3_2})  and (\ref{t_3_r2}) imply that 
 \begin{equation}\label{t_3_2_r2}
 t^2_{r+1}<t^{3}_{r+2}=\dots=t^3_{n_{{\mu},r}}< t^3_2=t^3_3=\dots=t^3_{r+1}
\end{equation} 
 
 Let $t^2_1$ be as given in (\ref{agent_1_next}). It follows from the definitions of $n_{r_1}$ and $n_{{\mu},r}$ in (\ref{define_n_u_r}) that $\dfrac{7}{n_{{\mu},r}-1} < \dfrac{5r+3}{2r}$.
  Then, (\ref{agent_1_next}) and (\ref{t_3_2})  together lead to
  \begin{equation}\label{t_3_2_1}
 t^3_2=t^3_3=\dots=t^3_{r+1} < t^2_1
  \end{equation}
  
 \noindent \textit{$4$) Computation of $x_j\big(t^3_j\big)$ for $j\geq r+2$}

Consider any agent $a_j,~j= r+2,\dots,n_{{\mu},r}$. Recall that $|z_j(t)| \leq \alpha,\forall t \geq 0$ and $n_j=n_{{\mu},r}-2$. Then,  Lemma \ref{average_inside}.$1$ leads to
\begin{equation*}
x_j\big(t^3_j\big) = \dfrac{\sum_{k \in S_j} ~x_k\big(t^2_j\big) }{n_{{\mu},r}-2} = \dfrac{\sum_{k=2,k\neq j}^{n_{{\mu},r}} ~x_k\big(t^2_j\big) }{n_{{\mu},r}-2} 
\end{equation*}
It follows from (\ref{x_l_less_u}), (\ref{x_j_second}) and (\ref{t_2_order}) that
$(5-\widetilde{\mu}) < x_k\big(t^2_j\big) \leq 5,~\forall k=2,\dots,n_{{\mu},r}$. Hence, the average of $x_k\big(t^2_j\big)$'s, i.e., $x_j\big(t^3_j\big)$, satisfies
$(5-\widetilde{\mu}) < x_j\big(t^3_j\big) \leq 5$.
Recall that $\dot{x}_j(t)=u_j(t)$ is constant in the interval $[t^2_j,t^3_j]$. Recall also that $\widetilde{\mu}=\dfrac{\mu}{2}$. Hence,
\begin{equation}\label{x_r2_3}
5- \mu < 5-\widetilde{\mu} < x_j(t) \leq 5,~\forall t \in \big[t^2_j,t^3_j\big],~\forall j=r+2,\dots,n_{{\mu},r}
\end{equation}

\noindent \textit{$5$) Computation of $x_l\big(t^3_l\big)$ for $l=2,\dots,r+1$}

Consider any agent $a_l,~l=2,\dots,r+1$.
Recall from (\ref{z_l_2_less}) that
  $|z_l\big(t^2_l\big)| \leq \alpha$. Recall also that $n_l=n_{{\mu},r}-1$. Then, Lemma \ref{average_inside}.$1$ leads to,
\begin{equation}\label{x_j_t_2_2}
x_l\big(t^3_l\big) = \dfrac{x_1\big(t^2_l\big)}{n_{{\mu},r}-1} + \dfrac{\sum_{k=2,k\neq l}^{n_{{\mu},r}} x_k\big(t^2_l\big) }{n_{{\mu},r}-1} 
\end{equation}
Recall from (\ref{x_l_less_u}) that $(5-\widetilde{\mu}) < x_k\big(t^2_l\big) < 5,~\forall k=2,\dots, r+1$. Further, it follows from (\ref{t_2_order}), (\ref{t_3_2_r2}) and (\ref{x_r2_3}) that
$(5-\widetilde{\mu}) < x_k\big(t^2_l\big) \leq 5,~\forall k=r+2,\dots,n_{{\mu},r}$.
Then, the definitions of $n_{{\mu}_2}$, $n_{{\mu}, r}$ and $\widetilde{\mu}$ lead to 
$(5-\mu)< x_l\big(t^3_l\big) \leq 5$ for $l=2,\dots,r+1$.
Recall that $\dot{x}_l(t)=u_l(t)$ is constant in the interval $[t^2_l,t^3_l\big]$. Hence, 
\begin{equation}\label{x_2_3}
5-\mu < x_l(t) \leq 5,~~\forall t \in \big[t^2_l,t^3_l\big],~~\forall l=2,\dots,r+1
\end{equation}

Now, by repeating the arguments which lead to (\ref{x_r2_3}) and (\ref{x_2_3}), we can show that
\begin{equation*}\label{x_all}
5-\mu < x_i(t) \leq 5,~~\forall t \geq t^2_{r+2}=\dfrac{3}{n_{{\mu}_r}-2},~~\forall i=2,\dots,n_{{\mu},r}
\end{equation*}
It follows from the definitions of $n_{{\mu}_1}$ and $n_{{\mu},r}$ in (\ref{define_n_u_r}) that $\mu \geq t^2_{r+2}$.
Hence, $5-\mu < x_i(t) \leq 5,~\forall t \geq \mu$ for $i=2,\dots,n_{{\mu},r}
$.
This proves the claim (\ref{claim}) and completes the proof.
\end{proof}

\subsection*{D.~Proof of Theorem \ref{strict_theo_1}}

Let $\epsilon>0$ be the given real number. Define $\widehat{\mu}:=\dfrac{\epsilon}{2}$ and 
$\widehat{r}:=\max\bigg\{ceil\bigg(\dfrac{2\alpha}{\epsilon}\bigg), ceil\bigg(\dfrac{\alpha}{5-\epsilon}\bigg) \bigg\}$.
Let $n_{\widehat{\mu},\widehat{r}}$ be as
defined in (\ref{define_n_u_r}) corresponding
to $\widehat{\mu}$ and $\widehat{r}$. Recall from (\ref{z_2_2_2}) that $|z_l(t)| \leq \alpha,~\forall t \geq t^2_{r+2},~\forall l=2,\dots,r+1$, where
the expression of $t^2_{r+2}$ is given in (\ref{agent_r2_next}). Then, it follows from (\ref{agent_r2_next}), the definitions of $n_{{\mu}_1}$ and $n_{{\mu},r}$ in (\ref{define_n_u_r}) and the definition of $n_{\widehat{\mu},\widehat{r}}$ that $\widehat{\mu} \geq t^2_{r+2}$. Thus,
\begin{equation}\label{z_inside_2_mu}
|z_i(t)| \leq \alpha,~~\forall t \geq \widehat{\mu},~~\forall i=2,\dots,r+1
\end{equation}
Further, recall from (\ref{z_r2}) that 
\begin{equation}\label{z_inside_r2_mu}
|z_i(t)| \leq \alpha,~~\forall t \geq 0,~~\forall i=r+2,\dots,n_{\widehat{\mu},\widehat{r}}
\end{equation}
Then, (\ref{z_inside_2_mu}) and (\ref{z_inside_r2_mu}) imply that the $\alpha$-consensus time is equal to the maximum of ${\widehat{\mu}}$ and the time required to steer $z_1$ inside $[-\alpha,\alpha]$. 

Let $\widehat{t}_1$ be the first time instant at which $|z_1\big(\hspace*{0.03cm}\widehat{t}_1\big)|=\alpha$. 
Recall that the initial conditions of the agents in Example \ref{strict_bound} are $x_1(0)=0$ and $x_i(0)=5,~\forall i=2,\dots,n_{\widehat{\mu},\widehat{r}}$.
Thus,
$z_1(0)=-5r<-\alpha=-3$.
Then, clearly, $z_1\big(\hspace*{0.03cm}\widehat{t}_1\big)  = -\alpha$. We know from Lemma \ref{x_i_close_5} that  
$x_i(t) \geq (5-\widehat{u}),~\forall t\geq \widehat{\mu},~\forall i=2,\dots,n_{\widehat{\mu},\widehat{r}}$.
Then, $z_1\big(\hspace*{0.03cm}\widehat{t}_1\big)  = -\alpha$ implies that 
 $x_1\big(\hspace*{0.03cm}\widehat{t}_1\big) \geq \bigg(5-\widehat{\mu}-\dfrac{\alpha}{\widehat{r}}\bigg)$. Further,
 the control rule (\ref{control_out})
 and $z_1\big(\hspace*{0.03cm}\widehat{t}_1\big)  = -\alpha$ together imply that
$\dot{x}_1(t) =\beta=1,~\forall t \in \big[0,\widehat{t}_1\big)$.
Recall that $x_1(0)=0$. Hence, $\widehat{t}_1 \geq \bigg(5-\widehat{\mu}-\dfrac{\alpha}{\widehat{r}}\bigg)$. Then, it follows from the definitions of $\widehat{\mu}$ and $\widehat{r}$ that 
$\widehat{t}_1 \geq (5-\epsilon)$ and $\widehat{t}_1 \geq \widehat{\mu}$. Subsequently, (\ref{z_inside_2_mu}) and (\ref{z_inside_r2_mu}) imply that the $\alpha$-consensus time $T\big(\alpha,X(0)\big)$ satisfies
$T\big(\alpha,X(0)\big)  \geq (5-\epsilon) = \Big(2T^{\ast}\big(X(0)\big)-\epsilon\Big)$.
This proves the first inequality in
(\ref{bound_T_alpha}).

The second inequality in (\ref{bound_T_alpha}) follows from Theorem \ref{consensus_time}. This completes the proof.

\subsection*{E.~Proof of Lemma \ref{opt_inside}}
It is given that $|z_i(t^k_i)| \leq \alpha$. Then, it follows from Lemma \ref{average_inside}.$2$ that control rule (\ref{update_instant})-(\ref{control}) satisfies
constraint (\ref{constraint}).
Let $t^{k+1}_i$ and $u^{\ast}_i$ be as defined in (\ref{update_instant})
and (\ref{control}), respectively.
 Then, the inter-update
duration $\big(t^{k+1}_i-t^k_i\big)$ is $\dfrac{\alpha}{\beta n_i}$, for all $k\geq 1$.
Next, we show that $\dfrac{\alpha}{\beta n_i}$ is the max-min value of 
optimization (\ref{max_min_con})-(\ref{constraint}). 
 This will prove the optimality of control rule (\ref{update_instant})-(\ref{control}) and complete the proof.
 
 For the sake of contradiction, assume that the max-min
 value of optimization (\ref{max_min_con})-(\ref{constraint}) is greater than
 $\dfrac{\alpha}{\beta n_i}$.
 Let that value be $T_i:=\big(t^{k+1}_i - t^k_i\big) > \dfrac{\alpha}{\beta n_i}$. Let $\widetilde{u}_i\in \mathcal{U}$ be the corresponding  optimal control.
 Then, it follows from the evolution of $z_i$ given in (\ref{z_trajectory}) that
 \begin{equation}\label{z_in_proof_4}
  z_i\big(t^{k+1}_i\big)=z_i(t^k_i) + n_i \int_{t^k_i}^{t^k_i+T_i}\!\!\!\!\!\!\widetilde{u}_i(t)~ dt - \sum_{j\in S_i} \Bigg ( \int_{t^k_i}^{t^k_i+T_i}\!\!\!\!\!\!\!u_j(t)~ dt \Bigg )
 \end{equation}
 Define $\widetilde{z}_i\big(t^{k+1}_i\big):= z_i\big(t^k_i\big) + n_i \int_{t^k_i}^{t^k_i+T_i}\widetilde{u}_i(t)~ dt$.
 Then, depending on the value of $\widetilde{z}_i\big(t^{k+1}_i\big)$, there are following three cases.\\
 \text{Case 1 : } \big|$\widetilde{z}_i\big(t^{k+1}_i\big)\big| > \alpha$\\
 Take $u_j(t)=0,~\forall t\in [t^k_i,t^k_i+T_i],~\forall j\in S_i$. Then, it follows from (\ref{z_in_proof_4}) that
 $\big|z_i\big(t^{k+1}_i\big)\big| > \alpha$.\\
 \text{Case 2 : } $0\leq \widetilde{z}_i\big(t^{k+1}_i\big) \leq \alpha$\\
 Take $u_j(t)=-\beta,~\forall t\in [t^k_i,t^k_i+T_i],~\forall j\in S_i$. Then, it follows from (\ref{z_in_proof_4}) that
 \begin{equation*}
 z_i\big(t^{k+1}_i\big) = \widetilde{z}_i\big(t^{k+1}_i\big) + \sum_{j\in S_i}\beta\big(t^k_i+T_i-t^k_i\big) = \widetilde{z}_i\big(t^{k+1}_i\big) + \beta n_i T_i
 \end{equation*}
 Recall that $T_i > \dfrac{\alpha}{\beta n_i}$, which is equivalent to $\beta n_i T_i > \alpha$. Then, the fact $\widetilde{z}_i\big(t^{k+1}_i\big) \geq 0$ implies that  $\big|z_i\big(t^{k+1}_i\big)\big| > \alpha$.\\
 \text{Case 3 : } $-\alpha\leq \widetilde{z}_i\big(t^{k+1}_i\big) < 0$\\
 Take $u_j(t)=\beta,~\forall t\in [t^k_i,t^{k+1}_i],~\forall j\in S_i$. Then, by following the arguments in Case $2$, we get 
 $\big|z_i\big(t^{k+1}_i\big)\big| > \alpha$.
 
 In each of the above cases, for control $\widetilde{u}_i$, there exist $u_j \in \mathcal{U}, \forall j\in S_i$,
 such that $\big|z_i\big(t^{k+1}_i\big)\big| > \alpha$.
 This violates constraint (\ref{constraint}) and in result, contradicts the assumption that 
 $\widetilde{u}_i$ is a solution of optimization  (\ref{max_min_con})-(\ref{constraint}).
 This contradiction proves our claim that
 $\dfrac{\alpha}{\beta n_i}$ is the max-min value of optimization (\ref{max_min_con})-(\ref{constraint}) and completes the proof.

\subsection*{F.~Proof of Lemma \ref{opt_outside}}
We first show that control rule (\ref{update_instant_out})-(\ref{control_out}) satisfies constraint (\ref{constraint_out}).
Let $t^{k+1}_i$ and $u^{\ast}_i$ be as defined in (\ref{update_instant_out}) and (\ref{control_out}),  respectively. 
Define $T_i:= t^{k+1}_i - t^k_i= \dfrac{|z_i\big(t^k_i\big)|+\alpha}{2\beta n_i}$.
It is given that $z_i\big(t^k_i\big) < -\alpha$. Then, (\ref{control_out}) leads to
$u^{\ast}_i(t) = \beta,~\forall t \in [t^k_i,t^k_i+T_i]$. Recall the evolution of $z_i$ given in (\ref{z_trajectory}). By putting the expressions of $t^{k+1}_i$ and $u^{\ast}_i$ in (\ref{z_trajectory}), we get

 \begin{equation}\label{z_in_proof_1_out}
  z_i\big(t^{k+1}_i\big)=z_i\big(t^k_i\big) + \dfrac{ \big|z_i\big(t^k_i\big)\big|+\alpha}{2}    -
  \sum_{j\in S_i} \Bigg (
  \int_{t^k_i}^{t^k_i + T_i} \!\!u_j(t) ~ dt\Bigg)
 \end{equation}
 Recall that the cardinality of $S_i$ is $n_i$ and
 $|u_j(t)| \leq \beta,~\forall t\geq 0,~\forall j\in P$.
  Then, it follows from the definition of $T_i$ that
 \begin{equation}\label{u_j_term}
  \Bigg | \sum_{j\in S_i} \Bigg( \int_{t^k_i}^{t^k_i + T_i}\! \! u_j(t) ~ dt \Bigg)\Bigg | \leq \dfrac{ \big|z_i\big(t^k_i\big)\big|+\alpha}{2}
 \end{equation}
Now, (\ref{z_in_proof_1_out}) and (\ref{u_j_term}) together imply that
$z_i\big(t^{k+1}_i\big) \leq \Big( z_i\big(t^k_i\big) + |z_i\big(t^k_i\big)|+\alpha \Big)$. As $z_i\big(t^k_i\big) < 0$, i.e., $z_i\big(t^k_i\big) + |z_i\big(t^k_i\big)|=0$,  we have
$z_i\big(t^{k+1}_i\big) \leq \alpha$.
Further, it follows from the definition of $z_i$ in (\ref{local_dis}) that
$
\dot{z}_i(t) = \Big(n_i u^{\ast}_i(t) - \sum_{j \in S_i} u_j(t)\Big) = \Big(n_i \beta - \sum_{j \in S_i} u_j(t)\Big),~\forall t \in [t^k_i,t^{k+1}_i]$.
Recall again that the cardinality of $S_i$ is $n_i$ and
 $|u_j(t)| \leq \beta,~\forall t\geq 0,~\forall j\in P$.
Hence, 
$\dot{z}_i(t) \geq 0 ,~\forall t \in [t^k_i,t^{k+1}_i]$ which implies that
$z_i(t) \leq z_i\big(t^{k+1}_i\big) ,~\forall t \in [t^k_i,t^{k+1}_i]$.
Then, the above-proved fact $z_i\big(t^{k+1}_i\big) \leq \alpha$ implies that control law (\ref{update_instant_out})-(\ref{control_out})
satisfies constraint (\ref{constraint_out}).

Next, we show that control rule (\ref{update_instant_out})-(\ref{control_out}) achieves
the max-min value of optimization (\ref{max_min_time_out})-(\ref{constraint_out}).
Recall that under control rule (\ref{update_instant_out})-(\ref{control_out}), the inter-update duration is $T_i= \dfrac{|z_i\big(t^k_i\big)|+\alpha}{2\beta n_i}$. We claim that $T_i$ is the max-min value of
optimization (\ref{max_min_time_out})-(\ref{constraint_out}). 
For the sake of contradiction, assume that the max-min value of optimization 
(\ref{max_min_time_out})-(\ref{constraint_out}) is $(T_i+\epsilon)$ for some $\epsilon > 0$.
Then, $\big(t^{k+1}_i-t^k_i\big)=T_i+\epsilon$. Let the inputs of the neighbouring agents be
$u_j(t) = -\beta,~\forall t\in [t^k_i,t^{k+1}_i),~\forall j \in S_i$.
Then, it follows from the evolution of $z_i$ given in (\ref{z_trajectory}) that
\begin{equation*}
  z_i\big(t^{k+1}_i\big)=z_i\big(t^k_i\big) + n_i \int_{t^k_i}^{t^k_i+T_i+\epsilon}\!\!\!\!\!\! \beta~ dt + \sum_{j\in S_i} \Bigg ( \int_{t^k_i}^{t^k_i+T_i+\epsilon}\beta ~ dt \Bigg )
 \end{equation*}
Now, by putting the expression of $T_i$ and rearranging the terms, we get
$z_i\big(t^{k+1}_i\big)=\Big(z_i\big(t^k_i\big) + |z_i\big(t^k_i\big)| + \alpha + 2\beta n_i \epsilon\Big)$. As $z_i\big(t^k_i\big)<0$,
we have $z_i\big(t^k_i\big) + |z_i\big(t^k_i\big)|=0$.
Thus, $ z_i\big(t^{k+1}_i\big)=\big( \alpha + 2\beta n_i \epsilon \big)>  \alpha$.
This violates constraint (\ref{constraint_out}) and contradicts 
 the fact that $(T_i+\epsilon)$ is the max-min value of optimization (\ref{max_min_time_out})-(\ref{constraint_out}).
Hence, the 
inter-update duration $T_i$ under control rule (\ref{update_instant_out})-(\ref{control_out})
is the 
max-min value of optimization (\ref{max_min_time_out})-(\ref{constraint_out}). This completes the proof.

\vspace*{-0.7cm}
\begin{IEEEbiographynophoto}
{Vishal Sawant}
received the Bachelor of Engineering degree
from Mumbai University, India, in 2010 and
the M.Tech. degree from the Department of Electrical Engineering, IIT Bombay, India, in 2012.
He then worked as a software engineer for two years.
In July 2014, he started Ph.D. degree at the Department of Electrical Engineering, IIT Bombay, India. From May 2020  to September 2021, he worked as a postdoctoral researcher at 
the Cognitive
Robotics group, TU Delft, The
Netherlands. 
Currently, he is a postdoctoral researcher at the Automation and Control Section, Aalborg University, Denmark.
His research interests include multi-agent systems, cyber-physical security and distributed optimization.
\end{IEEEbiographynophoto}
\vspace*{-0.8cm}
 \begin{IEEEbiographynophoto}
{Debraj Chakraborty} received the
Bachelor’s degree from Jadhavpur University,
Kolkata, India, in 2001 and the M. Tech. degree from Indian Institute of Technology Kanpur,
Kanpur, India, in 2003, both in electrical engineering, and the Ph. D. degree in electrical and
computer engineering from University of Florida,
Gainesville, FL, USA, in 2007.
He joined the Indian Institute of Technology
Bombay, Mumbai, India, in 2007, where he is
currently a Professor in the Department of Electrical Engineering. His research interests include optimal
control, linear systems, and multiagent systems.
\end{IEEEbiographynophoto} 
\vspace*{-0.7cm}
\begin{IEEEbiographynophoto}
{Debasattam Pal} received his Bachelor of Engineering (B.E.) degree
from the Department of Electrical Engineering of Jadavpur University,
Kolkata, in 2005. He received his M.Tech. and Ph.D. degrees from the
Department of Electrical Engineering, IIT Bombay, in the years 2007
and 2012, respectively. He then worked as an Assistant Professor in
IIT Guwahati from July, 2012 to May, 2014. He joined IIT Bombay in June, 2014, where he is currently an Associate Professor in the EE Department.
His main area of research is systems and control theory. More
specifically, his areas of interest are: multidimensional systems theory,
algebraic analysis of systems, dissipative systems, optimal control and
computational algebra.
\end{IEEEbiographynophoto} 

\end{document}